%% file: main.tex
\documentclass[11pt]{article}

\usepackage[T1]{fontenc}
\usepackage[utf8]{inputenc}
\usepackage{amsmath,amssymb,amsthm,mathtools}
\usepackage{mathpartir}
\usepackage{stmaryrd}
\usepackage[numbers]{natbib}
\usepackage{xcolor}
\usepackage[colorlinks=true,linkcolor=blue!50!black,citecolor=green!40!black,urlcolor=blue!50!black]{hyperref}
\usepackage{microtype}
\usepackage{enumitem}
\usepackage{graphicx}
\usepackage{booktabs}

\theoremstyle{plain}
\newtheorem{theorem}{Theorem}[section]
\newtheorem{proposition}[theorem]{Proposition}
\newtheorem{lemma}[theorem]{Lemma}
\newtheorem{corollary}[theorem]{Corollary}

\theoremstyle{definition}
\newtheorem{definition}[theorem]{Definition}
\newtheorem{example}[theorem]{Example}
\newtheorem{remark}[theorem]{Remark}

\newcommand{\type}{\ \mathsf{type}}
\renewcommand{\hom}[3]{\mathsf{hom}_{#1}(#2,#3)}
\newcommand{\homx}[2]{\mathsf{hom}(#1,#2)}
\newcommand{\Tri}{\mathsf{Tri}}
\newcommand{\Fill}{\mathsf{Fill}}
\newcommand{\Canon}{\mathsf{Canon}}
\newcommand{\Fork}{\mathsf{Fork}}
\newcommand{\Sense}{\mathsf{Sense}}
\newcommand{\Refr}{\mathsf{Ref}}
\newcommand{\apart}{\mathrel{\#}}
\newcommand{\simV}[1]{\mathrel{\simeq_{#1}}}
\newcommand{\apV}[1]{\mathrel{\#_{#1}}}
\newcommand{\MC}{\Delta}

\newcommand{\refl}{\mathsf{refl}}
\newcommand{\isContr}{\mathsf{isContr}}

\newcommand{\emp}{\mathbf{0}}
\newcommand{\transp}{\mathsf{tr}}
\newcommand{\Sem}[1]{\llbracket #1 \rrbracket}

\title{A Type Theory of Sense:\\ \large Witnessed Choice in Stratified Semantic Spaces}

\author{Iman Poernomo}
\date{June 2026}

\begin{document}
\maketitle

\begin{abstract}
We present a dependent type theory, \textbf{TTS}, in which Fregean
\emph{sense} is a mathematical object and semantic verdicts are
\emph{consequences of measurement}. The theory is motivated by a recent
empirical discovery: the token subspaces of large language models are not
manifolds. Semantic space is stratified, and a trajectory of meaning
through it encounters loci at which composition succeeds in
\emph{multiple inequivalent ways}. Classical homotopy type theory cannot
see these loci: its Kan condition guarantees that composites exist and
are essentially unique. TTS replaces the global Kan assumption with a
local analysis of \emph{filler spaces}, graded by \emph{regimes} ---
instruments at declared resolutions, ordered by refinement. Two
constructive relation formers carry the logic: regime-indexed
indiscernibility, and regime-indexed apartness in the tradition of
Brouwer, Heyting and Bishop. Apartness has no premise-free introduction
rule; separation witnesses enter the theory only through a
\emph{measurement context} --- a record of constants interpreted by an
actual instrument run --- so that TTS is, by construction, a logic of
consequence from measurement data. Two grades of compositional demand
are derived, not primitive --- \emph{canonicity} (the instrument cannot
distinguish any two fillers: composition is forced, no semantic decision
occurs) and the \emph{fork} (two warranted fillers, recorded apart:
composition demands a genuine choice) --- with a third, \emph{locked},
derived from them. Their mutual exclusion is a theorem; that fork
verdicts persist under refinement while canonicity does not are
theorems; and the calculus is conservative over its base, derives no
fork from the empty record, and derives nothing the record does not
force (provenance). \emph{Sense} is then a choice of
filler; \emph{reference} is the horn boundary that constrains the
choice; hyperintensional distinctness is apartness of fillers at a
regime; and the consistency of an identity discovered by one instrument
with a separation held by another is characterized exactly by the order
relation between the instruments --- a law instantiated twice, by
lexical ambiguity at filler spaces and by the name-state configuration
of Frege's own celestial case. We give the calculus, a semantics in
regime-indexed simplicial apartness spaces with soundness, an account of
basins and over-coherence, and a falsifiable bridge to two measurable
phenomena: manifold-hypothesis violation in token embedding spaces and
forking tokens in autoregressive generation.
\end{abstract}

\section{Introduction}\label{sec:intro}

\subsection{Meaning after the manifold}

A large language model in operation traverses a learned space. A
context is a position; a continuation is a directed step; an
unfolding conversation is a trajectory. Whatever ``meaning'' is for
such a system, it is enacted in the geometry of these traversals.
For a decade the default picture of that geometry was the
\emph{manifold hypothesis}: representations live on a smooth
low-dimensional submanifold of the embedding space, so that locally
there is always a well-behaved, essentially unique way to continue.

The default picture is now false as a matter of measurement. Robinson,
Dey and Chiang showed that the token subspaces of GPT-2, Llemma-7B,
Mistral-7B and Pythia-6.9B reject the manifold hypothesis --- and, for
three of the four, the weaker fiber-bundle hypothesis --- at
significance $p < 3\times10^{-7}$, with some token neighbourhoods
possessing no well-defined local dimension at all
\citep{robinson2025manifold}; an earlier study of the same group found
the token subspace to be a \emph{stratified} space, with strata of
varying dimension and curvature \citep{robinson2024structure}, and
independent probes confirm strongly varying local intrinsic dimension
in contextual embeddings \citep{ruppik2025lid, li2025stratified}.
Complementarily, Bigelow et al.\ exhibited \emph{forking tokens} in
autoregressive generation: positions at which resampling a single
token flips the entire downstream outcome distribution between
discrete alternatives \citep{bigelow2025forking}. Generation, observed
from outside, passes through points of genuine choice.

This paper takes the measurements seriously as constraints on
\emph{logic}. If semantic space is stratified rather than smooth, then
a formal theory of meaning for such spaces must be able to express,
within its formal apparatus, the three local \emph{grades} of
compositional demand a trajectory can encounter:

\begin{enumerate}[label=(\roman*)]
\item \textbf{Canonical.} Locally the space behaves like a
  manifold: a continuation exists and is essentially unique. No
  semantic decision occurs. (The manifold hypothesis holds here; in
  simplicial language, inner horns fill uniquely --- the Segal
  condition \citep{joyal2002, riehlshulman2017}.)
\item \textbf{Forked.} Composition succeeds in several
  \emph{inequivalent} ways: the space of fillers is inhabited but
  disconnected. The trajectory must choose, and nothing in the local
  geometry makes the choice for it. (The manifold hypothesis fails;
  this is the grade the measurements above detect.)
\item \textbf{Locked.} Composition is canonical \emph{and}
  every exit is foreclosed: the trajectory is confined to a tightly
  self-confirming region. (Degenerate fluency; in language models,
  the phase of repetition and collapse \citep{arnold2024phase,
  nakaishi2026critical}.)
\end{enumerate}

A caution before the formalism, because the introduction has so far
spoken loosely. Manifold smoothness, Segal-style uniqueness of
composition, and connectedness of a measured completion distribution
are \emph{three different properties}; none follows from the others.
A smooth manifold contains indefinitely many paths --- unique
continuation needs dynamics, not just smoothness --- and
non-manifoldness does not by itself disconnect any particular filler
space. The honest statement of the bridge: non-manifold loci
\emph{motivate candidate sites} at which compositional canonicity
should be tested; whether a given filler space separates is measured
independently, by the fork test of \S\ref{sec:empirics}. The grades
above are the formal properties; the geometric glosses are the
intended picture in the models of \S\ref{sec:semantics}, not
identifications.

One more lesson of the measurements shapes the design. Every verdict
above is \emph{instrument-relative}: the geometry was measured by a
particular probe at a particular resolution, and a structure invisible
at one resolution is decisive at another. We give this relativity
first-class status: TTS judgments are graded by \emph{regimes} ---
instruments at declared resolutions, ordered by refinement --- and the
laws relating verdicts across regimes (separations persist under
refinement; canonicity does not; two instruments' verdicts on one pair
can disagree exactly when the connecting instrument is not a refinement
of the separating one) are theorems of the system, not glosses.

Homotopy type theory (HoTT) \citep{hottbook} is constitutionally blind
to (ii) and (iii). Its types are Kan complexes: every horn fills, and
the space of fillers of an inner horn is contractible. Composition
never poses a question. This is precisely what makes HoTT an
extraordinary logic of \emph{mathematical} space --- and what makes it
the wrong logic of \emph{semantic} space, where the interesting events
are exactly the failures of canonicity.

\subsection{Sense as choice of filler}

The second motivation is older. Frege distinguished the
\emph{reference} of an expression from its \emph{sense} --- the mode
of presentation under which the reference is given --- to explain why
true identity statements can be informative \citep{frege1892}. The
distinction has resisted mathematization: possible-worlds intensions
are too coarse (logically equivalent contents collapse), while
syntactic individuations are too fine. The modern proposals ---
sense as algorithm \citep{moschovakis2006}, sense as procedure in a
ramified hierarchy \citep{duzi2010}, sense as program with
computational identity \citep{bentzen2020}, sense of a proof as its
derivational structure \citep{tranchini2016, ayhan2021, ayhan2025} ---
share a conviction we endorse: sense is something \emph{enacted}, a
way of arriving, not a static intension. What they lack is a geometry
in which ``ways of arriving'' have intrinsic structure: a space of
ways, with its own connectivity, whose disconnectedness is precisely
hyperintensional difference.

TTS supplies that geometry, and the stratified picture of semantic
space supplies its empirical referent. The central identification of
this paper:

\begin{quote}
\emph{The reference of a compositional demand is its boundary --- the
horn: what is already given and must be respected. A sense is an
inhabitant of the filler space of that horn --- a way of completing
what is given. Hyperintensional distinctness is apartness of fillers.
An identity statement is informative because a path between apart-held
fillers must be constructed and witnessed; it is never given for
free.}
\end{quote}

Where the filler space is canonical (grade (i)), sense and
reference coincide: there is only one way of arriving, and Frege's
distinction has no purchase --- which is why it is invisible to any
logic that assumes the Kan condition. The distinction \emph{is} the
fork. Frege's puzzle lives exactly where the manifold hypothesis
fails.

\subsection{The grounding requirement}

A logic that asserts ``these two ways of meaning are distinct'' owes
an account of what \emph{witnesses} the distinctness. Negative
formulations --- a primitive judgment of non-coherence, or an
unprovability claim --- make the witness endogenous: the theory
stipulates the failure rather than exhibiting it, and the judgment
floats free of any mathematics that could certify it. We regard this
as a hard adequacy condition, and we meet it twice over. First, with a
venerable constructive instrument: an \emph{apartness} structure in the
tradition of Brouwer, Heyting and Bishop \citep{heyting1956,
bishopbridges1985, troelstravandalen1988}. Apartness is positive:
$f \apV{V} g$ is inhabited by a separation witness, subject to
irreflexivity, symmetry and cotransitivity, and indexed by the regime
$V$ that achieved the separation. Second --- and this is the
architectural step --- apartness has \emph{no premise-free introduction
rule}: separation witnesses enter the theory only as the constants of a
\emph{measurement context} $\MC$, the record of an actual instrument
run. The judgment form of TTS is $\Gamma \vdash_{\MC} \mathcal{J}$:
hypotheses on the left, the record below. Nothing in the theory is
stipulated to fail; failures of canonicity are exhibited by data, and
the slogan is enforced by the metatheory --- \emph{variables
hypothesize, constants record}.

Three consequences distinguish the resulting theory from any system
with primitive negative judgments. First, the mutual exclusion of
canonicity and fork is a \emph{theorem} (Theorem~\ref{thm:exclusion}):
indiscernibility transports apartness into self-apartness, which is
absurd. The logical architecture polices itself. Second, the calculus
provably cannot manufacture its own verdicts: it is conservative over
its base, derives no fork from the empty record
(Corollary~\ref{cor:noforkexnihilo}), and every derivable fork holds in
every model realizing the record (Theorem~\ref{thm:provenance}). TTS
is, by construction and by theorem, a \textbf{logic of consequence from
measurement data}: the world supplies the verdicts; the logic audits
what they force. Third, the theory needs no primitive judgment of
\emph{emptiness} at all. In the intended models generation is total ---
there is always a next token --- so the constructively meaningful
failure of the Kan condition is never ``no filler'' but ``too many,
inequivalently.'' Absence, where needed, is the ordinary derived
negation $\Fill(\sigma)\to\emp$; the load-bearing new judgment is the
fork.

\subsection{Contributions and plan}

\begin{enumerate}[label=(\arabic*)]
\item A dependent type theory, TTS, extending Martin-L\"of type theory
  with directed hom- and triangle types, regime-indexed
  indiscernibility and apartness, and measurement contexts, presented
  calculus-first with proof terms (\S\ref{sec:calculus}).
\item The derived grades of canonicity and fork; the exclusion theorem;
  the persistence asymmetry (forks persist under refinement, canonicity
  does not); and the cross-regime consistency theorem characterizing
  exactly when instruments may disagree (\S\ref{sec:grades}).
\item A formal reconstruction of sense, reference, hyperintensionality
  and the informativeness of identity, with lexical ambiguity as a
  worked example --- and the names configuration (name-states in a
  presentation space) falling under the same cross-regime law: two
  Fregean phenomena, one theorem, two typing instantiations
  (\S\ref{sec:sense}).
\item A semantics in regime-indexed simplicial apartness spaces with
  soundness; metatheorems making the measurement architecture
  more than rhetoric --- conservativity, no-fork-ex-nihilo, provenance,
  and consistency of honest records; and the characterization of
  \emph{basins} as components whose internal demands are canonical,
  with \emph{locked} basins as the third grade
  (\S\ref{sec:semantics}).
\item A falsifiable bridge to measurement: apartness verdicts at
  pullback regimes as the formal content of resampling experiments,
  with manifold-hypothesis violations and forking tokens as the
  empirical anchors (\S\ref{sec:empirics}).
\end{enumerate}

Related work is discussed in \S\ref{sec:related}; open problems,
including completeness and mechanization, in \S\ref{sec:conclusion}.
One independence claim, stated up front because it allocates risk.
The empirical case is at the pilot stage and says so
(\S\ref{sec:empirics}); the philosophical and logical contributions
do not wait on it. That HoTT is constitutionally blind to
compositional choice, and that a calculus of witnessed, graded
verdicts can see it --- with exclusion, persistence-asymmetry, and
cross-instrument consistency as theorems --- stands even if semantic
spaces prove smoother than the present measurements suggest. Put
most plainly: what this paper builds is a small logic of
\emph{instrument-relative semantic disagreement}, in which
difference must be witnessed, identification need not erase
difference, and verdicts propagate asymmetrically under refinement.
Frege, stratified LLM geometry, and alignment-induced mode collapse
are three applications of that one machine.

A note on method: the paper is organized so that nothing is asked of
the reader on credit. An actual measurement, four data points small,
opens \S\ref{sec:calculus} (Example~\ref{ex:miniature}), and every
construct of the calculus is anchored to it as it is introduced; the
full experiment of \S\ref{sec:empirics} is that example at scale.

\section{The calculus TTS}\label{sec:calculus}

TTS is an extension of intensional Martin-L\"of type theory (MLTT) with
$\Pi$-, $\Sigma$-, identity, empty, unit, sum and propositional-truncation
types; we follow the conventions of \citep{hottbook}, writing $p : a =_A b$
for paths, $\refl$, and $\transp$ for transport. Four layers are added: a
\emph{semantic signature} and its types (\S\ref{ssec:semtypes}), directed
structure (\S\ref{ssec:hom}--\ref{ssec:tri}), \emph{regime-indexed}
indiscernibility and apartness (\S\ref{ssec:regimes}--\ref{ssec:apart}), and
\emph{measurement contexts} (\S\ref{ssec:mctx}). The slogan governing the
whole design, defended in \S\ref{ssec:mctx}: \textbf{variables hypothesize;
constants record.}

\paragraph{A running example.} Throughout, the reader is invited to hold one
concrete instance of every formal object. Let $A$ be a space of contextual
states of a discourse --- for a language model, embedded prompt states ---
and $a, b, c$ three such states. A continuation $f : \hom{A}{a}{b}$ is a
realized stretch of discourse carrying the context from $a$ to $b$: say, the
words arriving at an occurrence of \emph{bank}. A second continuation
$g : \hom{A}{b}{c}$ departs from that occurrence without disambiguating it.
The pair $(f,g)$ is a compositional demand: the passage must be readable
whole. Everything below is the logic of that situation.

And so that no construct of this section floats free of measurement, we
fix at the outset one \emph{actual} instrument run, in miniature, drawn
from the experiment reported in full in \S\ref{sec:empirics}. The
calculus will refer back to it constantly; nothing later in the paper is
needed to understand it.

\begin{example}[A measurement, in miniature]\label{ex:miniature}
A language model (\texttt{gpt-3.5-turbo-instruct}, temperature $1.0$) was
given the boundary
\begin{quote}
\emph{``We met them near the bank, just before sunset.''}
\end{quote}
Among its sampled continuations are these four (first sentences shown;
all data verbatim from the run):
\begin{itemize}[itemsep=1pt]
\item[$\varphi_1$:] \emph{``Corrapiran decided to keep watch outside,
  while I entered the small bank alone.''}
\item[$\varphi_2$:] \emph{``The bank was closed and the grounds were
  abandoned when we got there.''}
\item[$\varphi_3$:] \emph{``They were standing on the riverbank opposite
  from us, watching the water flow by.''}
\item[$\varphi_4$:] \emph{``As we paddled our kayaks up to the bank, we
  see that they were waiting for us.''}
\end{itemize}
The \emph{instrument}: embed each continuation
(\texttt{text-embedding-3-small}) and measure cosine distances. The
measured values:
\[
\begin{array}{c|cccc}
d_{\cos} & \varphi_1 & \varphi_2 & \varphi_3 & \varphi_4 \\\hline
\varphi_1 & 0 & .585 & .691 & .698 \\
\varphi_2 & & 0 & .720 & .652 \\
\varphi_3 & & & 0 & .455 \\
\varphi_4 & & & & 0
\end{array}
\]
Read at resolution $\varepsilon = 0.60$ --- call this instrument $V$ ---
the four continuations fall into exactly two chain-components:
$\{\varphi_1, \varphi_2\}$ (the financial readings, joined by their
$.585$ step) and $\{\varphi_3, \varphi_4\}$ (the riparian readings,
joined by $.455$), with every cross-distance exceeding $0.60$. The
instrument's \emph{record} of this run is the verdict list
\[
\MC_0 \;=\; \big(\;
\gamma_1 : \varphi_1 \simV{V} \varphi_2,\;\;
\gamma_2 : \varphi_3 \simV{V} \varphi_4,\;\;
\delta_1 : \varphi_1 \apV{V} \varphi_3,\;\;
\delta_2 : \varphi_1 \apV{V} \varphi_4,\;\dots\big)
\]
--- two connections, four separations. Now vary the instrument and watch
the verdicts move. At the \emph{finer} resolution $\varepsilon' = 0.40$
(call it $V'$, with $V' \le V$) every pairwise distance exceeds the
threshold: all four continuations separate; the $\delta$-entries survive,
the $\gamma$-entries do not recur. At the \emph{coarser} resolution
$\varepsilon'' = 0.70$ (call it $W$, with $V < W$) the four continuations
form a \emph{single} component: this instrument connects $\varphi_1$ to
$\varphi_3$ --- the very pair $V$ separates. Three facts of this little
dataset are, in miniature, the three load-bearing laws of the calculus
built in this section: separations persist under refinement while
connections do not (the monotonicity rules of \S\ref{ssec:apart}); one
instrument's verdicts on one pair never conflict (the clash lemma,
\S\ref{sec:grades}); and a coarse connection coexisting with a fine
separation is not a contradiction but a fact about \emph{two}
instruments (the cross-regime theorem, \S\ref{ssec:crossregime}). The
record $\MC_0$ is a measurement context in the sense of
\S\ref{ssec:mctx}; the fork verdict it licenses is
Definition~\ref{def:fork}; the full protocol --- calibrated boundaries,
$N = 300$, a persistence criterion --- is \S\ref{sec:empirics}.
\end{example}

\subsection{The semantic signature and semantic types}\label{ssec:semtypes}

\begin{definition}[Semantic signature]\label{def:semsig}
A \emph{semantic signature} $\Sigma$ comprises (i) a collection of
\emph{base semantic types} $\mathfrak{B}$, and (ii) typed \emph{term
constants}: names for points, and names for instrument-observed
continuations and fillers of declared horns.
\end{definition}

The term constants are not decoration: the closed-judgment metatheorems of
\S\ref{sec:semantics} concern closed terms, and over a constant-free
signature the directed fragment's only closed terms are degeneracies ---
the calculus would issue no fork verdict on any non-degenerate horn. The
record presupposes a naming discipline supplied by the instrument:
\emph{the instrument names what it separates} --- in
Example~\ref{ex:miniature}, the names $\varphi_1, \dots, \varphi_4$.

\begin{definition}[Semantic types]\label{def:semtype}
The \emph{semantic types} over $\Sigma$ form the grammar
\[
S \;::=\; \mathfrak{B}
  \;\mid\; \hom{S}{a}{b}
  \;\mid\; \Tri_S(f,g;h)
  \;\mid\; \textstyle\sum_{x:S} S'(x) \quad (S'(x) \text{ semantic}).
\]
Semantic types are a distinguished subclass of MLTT types: all MLTT
structural and eliminator rules apply to them (in particular
$\Sigma$-projections and $\Sigma$-elimination); only $\mathsf{hom}$ and
$\Tri$ are eliminator-free. $\Pi$-types, identity types, $\mathbf{0}$,
$\mathbf{1}$, $+$ and $\|\cdot\|$ are \emph{not} semantic. The relation
formers of \S\ref{ssec:apart} apply only to semantic types, over arbitrary
(possibly open) terms.
\end{definition}

Why this perimeter: semantic types are the types an instrument can
point at --- discourse states, continuations, warrants, and their
pairings. The needed distinction is between a function as an
\emph{object of} measurement and a function as an \emph{instrument
of} measurement: instruments point \emph{through} functions
constantly (every readout $\rho$ of \S\ref{ssec:readout} is one),
but they do not point \emph{at} them --- no record entry separates
two functions qua functions, and propositions likewise present
nothing for a resolution to resolve. The relation formers are
confined to the grammar for that reason; readouts live in the
signature, on the instrument side of the line.
``Eliminator-free'' means: no induction principle --- the theory can
\emph{name} continuations (via the signature's constants) but never
case-analyse them, a poverty that Remark~\ref{rem:canonunder} will show
is exactly what keeps canonicity honest.

\emph{Convention (substitution stability).} All new formers and constants of
this section commute with substitution. The path-action $e_* f$ of an
ambient path on a continuation is \emph{derived} (ordinary transport in
$x \mapsto \hom{A}{x}{b}$), not primitive: the directed fragment's
primitives are exactly the five rules displayed below.

\subsection{Directed hom-types}\label{ssec:hom}

For $A$ semantic and $a, b : A$, the type of \emph{continuations} from $a$
to $b$. Continuations are directed: no symmetry is posited --- a stretch of
discourse cannot be un-said --- and no composition is assumed: composition
is exactly what is at stake.

\begin{mathpar}
\inferrule[hom-form]
  {\Gamma \vdash_{\MC} A \type_{\mathsf{sem}} \\ \Gamma \vdash_{\MC} a : A \\ \Gamma \vdash_{\MC} b : A}
  {\Gamma \vdash_{\MC} \hom{A}{a}{b} \type_{\mathsf{sem}}}

\inferrule[hom-degen]
  {\Gamma \vdash_{\MC} a : A}
  {\Gamma \vdash_{\MC} \mathsf{id}_a : \hom{A}{a}{a}}
\end{mathpar}

In the running example an inhabitant of $\hom{A}{a}{b}$ is a realized
traversal; $\mathsf{id}_a$ is the null continuation --- remaining at $a$,
adding nothing.

\subsection{Triangle types and filler spaces}\label{ssec:tri}

A \emph{triangle} is a warrant that a candidate $h : \hom{A}{a}{c}$
genuinely composes $f$ and $g$: that reading the passage whole as $h$ is
compatible with having arrived by $f$ and departed by $g$. For a language
model the warrant is realizability: the model, holding both flanking
commitments, can enact $h$ as a fluent traversal.

\begin{mathpar}
\inferrule[tri-form]
  {\Gamma \vdash_{\MC} f : \hom{A}{a}{b} \\
   \Gamma \vdash_{\MC} g : \hom{A}{b}{c} \\
   \Gamma \vdash_{\MC} h : \hom{A}{a}{c}}
  {\Gamma \vdash_{\MC} \Tri_A(f,g;h) \type_{\mathsf{sem}}}

\inferrule[tri-degen-l]
  {\Gamma \vdash_{\MC} f : \hom{A}{a}{b}}
  {\Gamma \vdash_{\MC} \lambda^{\mathsf{l}}_f : \Tri_A(\mathsf{id}_a, f; f)}

\inferrule[tri-degen-r]
  {\Gamma \vdash_{\MC} f : \hom{A}{a}{b}}
  {\Gamma \vdash_{\MC} \lambda^{\mathsf{r}}_f : \Tri_A(f, \mathsf{id}_b; f)}
\end{mathpar}

(In \textsc{tri-form} and throughout, the presupposition
$A \type_{\mathsf{sem}}$ is carried by the premises' typing. The
degenerate triangles are the trivial warrants: prefixing or suffixing
the null continuation changes no reading, and the theory says so.)

\begin{definition}[Horn; filler space]\label{def:filler}
An \emph{inner $2$-horn} in a semantic type $S$ is a tuple
$\sigma := (a, b, c : S,\; f : \hom{S}{a}{b},\; g : \hom{S}{b}{c})$; its
\emph{boundary vertices} are $a, b, c$ (``vertex'' is horn-relative: for
$S$ itself a hom-type, the horn's vertices are ambient edges). Its
\emph{filler space} is
\[
\Fill_S(\sigma) \;:\equiv\; \sum_{h \,:\, \hom{S}{a}{c}} \Tri_S(f,g;h),
\]
semantic by the grammar; we write $\Fill(\sigma)$ when $S$ is clear.
\end{definition}

A filler is a \emph{pair}: a composite together with its compositional
warrant. Note what is not assumed: neither inhabitation (no Kan condition)
nor uniqueness of composition. Operationally, $\Fill_S(\sigma)$ is the
space of \emph{readings} --- everything the system can warrantedly do when
handed the boundary; sampling a model's continuations of a prompt is
sampling (a readout of) this space, a remark \S\ref{sec:empirics} makes
precise.

\begin{remark}[Why the warrant is binary]\label{rem:warrant}
$\Tri_A(f,g;h)$ is inhabited or not, while fluent enactment is
plainly graded: a bridge can be barely coherent or perfect. The
gradation is not lost; it is \emph{relocated}. The warrant is a
realizability condition --- can the system traverse the passage this
way at all --- and quality variation among warranted fillers is
exactly what the regime structure on $\Fill(\sigma)$ measures: a
barely-coherent bridge and a perfect one are both fillers, and any
instrument worth its resolution separates them. Grading the warrant
itself would duplicate, inside $\Tri$, the work the relations
already do outside it.
\end{remark}

\begin{remark}[Higher horns]\label{rem:higher}
We develop the theory for inner $2$-horns, where the semantic phenomena of
interest already occur. The generalization to $\Lambda^n_i$ is mechanical
but notationally heavy; nothing below depends on it.
\end{remark}

\subsection{Regimes}\label{ssec:regimes}

\begin{definition}[Regimes]\label{def:regime}
TTS is parameterized by a poset $(\mathcal{V}, \le)$ of \emph{regimes},
treated as meta-level indices (in the manner of universe levels), not
internal types. \textbf{Convention: $V' \le V$ means $V'$ is \emph{finer}}
--- it separates more and identifies less. In the intended models a regime
is an \emph{instrument}: a metric with a resolution, $V = (d, \varepsilon)$,
and finer means smaller $\varepsilon$. We write $V < W$ for the strict
order. The three instruments of Example~\ref{ex:miniature} are three
regimes: cosine distance read at $0.40$, $0.60$ and $0.70$ respectively,
ordered $V' \le V \le W$ --- one metric, three resolutions, three
verdict-sets.
\end{definition}

A regime is the formal home of an old hermeneutic commonplace: there is no
verdict of sameness or difference \emph{simpliciter}, only sameness or
difference \emph{as measured by an instrument at a resolution}. The
calculus never mentions numbers; it tracks only the order structure ---
which instrument refines which.

\subsection{Indiscernibility and apartness}\label{ssec:apart}

Two regime-indexed relation formers on each semantic type, both valued in
h-propositions (their isProp structure is built into the formers; where a
sum appears in a conclusion it is truncated).\footnote{For readers
outside type theory: an \emph{h-proposition} is a type with at most one
inhabitant up to identity --- a yes/no matter; the \emph{truncation}
$\|T\|$ records \emph{that} $T$ is inhabited while forgetting
\emph{which} inhabitant --- mere existence. The distinction earns its
keep at Definition~\ref{def:fork}.} In both formation rules below the
type $S$ is presupposed semantic ($S \type_{\mathsf{sem}}$), per
Definition~\ref{def:semtype}.

\paragraph{$V$-indiscernibility.} $u \simV{V} v$: \emph{the instrument $V$
cannot tell $u$ from $v$ apart} --- there is a chain of sub-resolution
deformations from one to the other. (The entries $\gamma_1, \gamma_2$ of
$\MC_0$ in Example~\ref{ex:miniature}.)

\begin{mathpar}
\inferrule[$\simeq$-form]
  {\Gamma \vdash_{\MC} u : S \\ \Gamma \vdash_{\MC} v : S}
  {\Gamma \vdash_{\MC} u \simV{V} v \type}

\inferrule[$\simeq$-refl]
  {\Gamma \vdash_{\MC} u : S}
  {\Gamma \vdash_{\MC} \mathsf{r}_V(u) : u \simV{V} u}

\inferrule[$\simeq$-sym]
  {\Gamma \vdash_{\MC} p : u \simV{V} v}
  {\Gamma \vdash_{\MC} p^{\dagger} : v \simV{V} u}

\inferrule[$\simeq$-trans]
  {\Gamma \vdash_{\MC} p : u \simV{V} v \\ \Gamma \vdash_{\MC} q : v \simV{V} w}
  {\Gamma \vdash_{\MC} p \cdot q : u \simV{V} w}

\inferrule[$\simeq$-mono]
  {\Gamma \vdash_{\MC} p : u \simV{V'} v \\ V' \le V}
  {\Gamma \vdash_{\MC} \mathsf{up}(p) : u \simV{V} v}
\end{mathpar}

\textsc{$\simeq$-mono} reads: what a finer instrument cannot distinguish, a
coarser one cannot either.

\begin{lemma}[id-to-regime]\label{lem:idtoregime}
For every regime $V$ there is a function
$\mathsf{itr}_V : (u =_S v) \to (u \simV{V} v)$.
\end{lemma}
\begin{proof}
Path induction with motive $(x, y, p) \mapsto x \simV{V} y$ (an h-prop
family over open terms, licensed by \textsc{$\simeq$-form}) and base
$\mathsf{r}_V$.
\end{proof}

Mathematical identity is thus indiscernible at every regime. The converse
is \emph{not} claimed --- indiscernibility at every declared regime does
not yield a path --- and this one-directionality is essential to what
follows.

\paragraph{$V$-apartness.} $u \apV{V} v$: \emph{the instrument $V$ has the
resources to separate $u$ from $v$} --- a positive structure in the
constructive-apartness tradition of Brouwer, Heyting and Bishop
\citep{heyting1956, bishopbridges1985, troelstravandalen1988,
vonplato1995}.

\begin{mathpar}
\inferrule[$\#$-form]
  {\Gamma \vdash_{\MC} u : S \\ \Gamma \vdash_{\MC} v : S}
  {\Gamma \vdash_{\MC} u \apV{V} v \type}

\inferrule[$\#$-irrefl]
  {\Gamma \vdash_{\MC} \delta : u \apV{V} u}
  {\Gamma \vdash_{\MC} \mathsf{absurd}(\delta) : \emp}

\inferrule[$\#$-sym]
  {\Gamma \vdash_{\MC} \delta : u \apV{V} v}
  {\Gamma \vdash_{\MC} \delta^{\dagger} : v \apV{V} u}

\inferrule[$\#$-cotrans]
  {\Gamma \vdash_{\MC} \delta : u \apV{V} v \\ \Gamma \vdash_{\MC} w : S}
  {\Gamma \vdash_{\MC} \mathsf{split}(\delta, w) :
    \big\| (u \apV{V} w) + (w \apV{V} v) \big\|}

\inferrule[$\#$-mono]
  {\Gamma \vdash_{\MC} \delta : u \apV{V} v \\ V' \le V}
  {\Gamma \vdash_{\MC} \mathsf{down}(\delta) : u \apV{V'} v}

\inferrule[$\#$-ext]
  {\Gamma \vdash_{\MC} p : u \simV{V} u' \\ \Gamma \vdash_{\MC} \delta : u \apV{V} v}
  {\Gamma \vdash_{\MC} \mathsf{ext}(p, \delta) : u' \apV{V} v}
\end{mathpar}

Glosses. \textsc{$\#$-mono}: a separation achieved by a coarse instrument
persists under every refinement --- one cannot lose a measured difference
by looking harder. (In Example~\ref{ex:miniature}: the $\delta$-entries
recorded at $0.60$ all hold at $0.40$; the $\gamma$-entries do not ---
which is why there is a \textsc{$\simeq$-mono} going coarse-ward and no
$\simeq$-rule going fine-ward.) \textsc{$\#$-ext}: apartness is
extensional over indiscernibility \emph{at the same regime} --- what an
instrument cannot distinguish, it cannot separate differently. \textsc{$\#$-cotrans}: if two
readings are separated, any third reading is separated from at least one of
them; it is cotransitivity that makes apartness classes behave like
components of a space. (Cotransitivity is retained as constitutive of
constructive apartness and holds in all intended models; no result in this
paper depends on it.)

\begin{remark}[Exogeneity, made architectural]\label{rem:exo}
There is \emph{no premise-free introduction rule} for apartness at any
regime: every $\#$ in a conclusion traces to a $\#$ in a premise
(\textsc{$\#$-cotrans} redistributes recorded apartness to new comparison
points; it cannot create apartness ex nihilo --- the erasure argument of
Appendix~\ref{app:soundness}, which also yields
Corollary~\ref{cor:noforkexnihilo}, makes this precise). Closed separation witnesses enter the theory in
exactly one way: through the measurement context of
\S\ref{ssec:mctx}, where they are constants interpreted by an actual
instrument record --- such as $\MC_0$ of Example~\ref{ex:miniature},
whose $\delta_1$ is not a stipulation but the number $.691$ read against
the threshold $0.60$. Distinctness claims are obligations to exhibit
data.
\end{remark}

\subsection{Measurement contexts}\label{ssec:mctx}

\begin{definition}[Measurement context]\label{def:mctx}
A \emph{measurement context} is a finite list
\[
\MC \;::=\; \emptyset
  \;\mid\; \MC,\, \delta : u \apV{V} v
  \;\mid\; \MC,\, \gamma : a \simV{V} b
\]
whose terms are \emph{closed} terms of semantic types and whose indices are
regimes (or pullback indices, \S\ref{ssec:readout}). $\MC$ is well-formed
($\MC\ \mathsf{ok}$) when each entry's terms are well-typed closed terms of
a single semantic type. The judgment form of TTS is
$\Gamma \vdash_{\MC} \mathcal{J}$, and the single rule by which the record
enters derivations is
\begin{mathpar}
\inferrule[$\MC$-ax]
  {(\delta : u \apV{V} v) \in \MC}
  {\Gamma \vdash_{\MC} \delta : u \apV{V} v}
\end{mathpar}
(presupposing $\Gamma\ \mathsf{ok}$, $\MC\ \mathsf{ok}$; dually for
$\gamma$-entries). $\MC$-entries are \emph{constants}: they cannot be
$\lambda$-abstracted or substituted for. Weakening of $\MC$ is admissible;
there is no $\MC$-strengthening.
\end{definition}

The intended reading: $\MC$ is the \emph{record} of an instrument run.
The reader has already held one in full: $\MC_0$ of
Example~\ref{ex:miniature} --- two connections, four separations, each
entry a number read against a declared threshold. ($\MC$-ax) is the rule
that lets a derivation cite $\delta_1$; nothing else about the run
enters the logic. TTS$(\Sigma, \MC)$ is then a \textbf{logic of
consequence from measurement data}: it derives what the record forces ---
which forks stand, which identities are refuted, which cross-instrument
configurations cohere --- and, by the metatheorems of
\S\ref{sec:semantics}, it can do nothing else: no fork is derivable from
the empty record (Corollary~\ref{cor:noforkexnihilo}), and every derivable
fork holds in every model realizing the record
(Theorem~\ref{thm:provenance}).

\begin{remark}[Variables hypothesize, constants record]\label{rem:slogan}
A separation hypothesized for the sake of argument is a $\Gamma$-variable:
it is $\Pi$-bindable, and it shows up in the \emph{type} of the resulting
term --- a conditional claim. A separation that has been measured is an
$\MC$-constant: it underwrites \emph{categorical} claims, and the
distinction is real exactly where the metatheorems of
\S\ref{sec:semantics} live --- closed theorems over a fixed record. We
also state it honestly: within one derivation a constant behaves like a
hypothesis, just as a theory's axioms behave like premises of its
proofs. The difference is that axioms cannot be discharged --- no $\Pi$
binds an $\MC$-entry, so no theorem internalizes a counterfactual
record. What certifies that the record itself is honest is
meta-theoretic and stated as such
(Theorem~\ref{thm:recordconsistency}).
\end{remark}
\section{Canonicity and the fork}\label{sec:grades}

Section~\ref{sec:calculus} supplied raw material and asserted nothing: hom-types
need not be inhabited, triangles need not exist, and apartness has no
premise-free introduction. The system so far cannot even say that composition
succeeds. The poverty is deliberate. The content of the logic lives in two
derived classifications of a filler space --- the \emph{grades} --- and in the
theorems governing their interaction across regimes.

\subsection{The clash lemma}

\begin{lemma}[Per-regime clash]\label{lem:path-apart}
For $S$ semantic, $u, v : S$ and any regime $V$:
\[
\mathsf{clash}_V : (u \simV{V} v) \to (u \apV{V} v) \to \emp .
\]
\end{lemma}
\begin{proof}
Given $p : u \simV{V} v$ and $\delta : u \apV{V} v$, apply
\textsc{$\#$-ext} instantiated at $(u, u', v) := (u, v, v)$ --- its premises
are exactly $p$ and $\delta$ --- to obtain $v \apV{V} v$; conclude by
\textsc{$\#$-irrefl}. The left-handed \textsc{$\#$-ext} suffices directly; no
symmetry shuffle is needed.
\end{proof}

In words: no single instrument may hold, at once, a deformation of one
reading into another and a separation between them. Across \emph{different}
instruments no such law holds --- that is Theorem~\ref{thm:crossregime},
and it is a feature.

\begin{corollary}[id-clash]\label{cor:idclash}
$(u =_S v) \to (u \apV{V} v) \to \emp$, by Lemma~\ref{lem:idtoregime} and
$\mathsf{clash}_V$. (Under the intended discrete interpretation of identity,
\S\ref{sec:semantics}: a recorded separation refutes literal pointwise
identity.)
\end{corollary}

\subsection{The grades, exclusion, and persistence}

\begin{definition}[Canonicity at a regime]\label{def:canon}
$\Canon_V(\sigma) \;:\equiv\; \big\|\Fill(\sigma)\big\| \times
\prod_{\varphi, \psi \,:\, \Fill(\sigma)} \varphi \simV{V} \psi$
--- an h-proposition: the filler space is inhabited and the instrument $V$
cannot tell any two of its inhabitants apart. \emph{The geometry decides;
no semantic decision occurs.} Operationally: every continuation the system
produces at $\sigma$ is, at this resolution, a deformation of every other
--- three hundred draws, one cluster, no news.
\end{definition}

\begin{definition}[Fork at a regime]\label{def:fork}
$\Fork_V(\sigma) \;:\equiv\; \sum_{\varphi, \psi \,:\, \Fill(\sigma)}
\varphi \apV{V} \psi$.
A fork witness is a triple --- two complete fillers, each a composite with
its warrant, and a recorded separation between them. $\Fork_V$ is
\emph{deliberately not truncated}: the verdict names its witnesses; the
truncation $\|\Fork_V(\sigma)\|$ is the judgment of record. (Exclusion and
persistence below use plain $\Sigma$-elimination and are insensitive to the
choice.) In the miniature: were $\varphi_1, \varphi_3$ the fillers of a
fixed horn, $(\varphi_1, \varphi_3, \delta_1)$ would be the paradigm
inhabitant --- two readings of one passage, measurably apart; the
precise gap between sampled completions and horn-fillers is the scoping
question \S\ref{sec:empirics} addresses.
\end{definition}

\begin{theorem}[Exclusion at a regime]\label{thm:exclusion}
$\Canon_V(\sigma) \to \Fork_V(\sigma) \to \emp$.
\end{theorem}
\begin{proof}
From $(i, c) : \Canon_V(\sigma)$ and $(\varphi, \psi, \delta) :
\Fork_V(\sigma)$: $c\,\varphi\,\psi : \varphi \simV{V} \psi$; apply
$\mathsf{clash}_V$ with $\delta$. (The truncated component $i$ is
discarded; no truncation-elimination into data occurs.)
\end{proof}

Exclusion remains a \emph{theorem}, not an axiom --- the architecture
polices itself --- and it is now indexed: it forbids a single instrument
from issuing both verdicts on one horn. It says nothing about two
instruments, by design.

\begin{theorem}[Fork persistence]\label{thm:persist}
For $V' \le V$: $\Fork_V(\sigma) \to \Fork_{V'}(\sigma)$, by
\textsc{$\#$-mono} on the third component. \emph{Persistence is
monotonicity}: a fork recorded by a coarse instrument stands under every
refinement.
\end{theorem}

\begin{proposition}[Canonicity is not persistent]\label{prop:asym}
There is a model of TTS (Definition~\ref{def:sas}) and a horn $\sigma$ with
$\Canon_{\varepsilon}(\sigma)$ and $\Fork_{\varepsilon'}(\sigma)$ for
resolutions $\varepsilon' < \varepsilon$.
\end{proposition}
\begin{proof}
Take a carrier with one composite $h : \homx{a}{c}$ and exactly two
warrants: $\Fill(\sigma) = \{(h, t_1), (h, t_2)\}$, with
$d(t_1, t_2) = r \in (\varepsilon', \varepsilon]$, every other point of the
ambient triangle level at distance $> \varepsilon$ from both, $f, g$
non-degenerate; complete to a model of Definition~\ref{def:sas} (discrete
identity; intrinsic component structures; the shared composite makes the
face conditions trivial). At $\varepsilon$: the single
$\le \varepsilon$ step joins the two fillers, so $\Canon_\varepsilon$
holds. At $\varepsilon'$: the two fillers lie in different
$\varepsilon'$-components, so $\Fork_{\varepsilon'}$ holds (and
$\Canon_{\varepsilon'}$ fails, by Exclusion).
\end{proof}

Note what the witnessing configuration is: \emph{two separated warrants
over one composite} --- the same arrival, two measurably distinct modes of
arriving. The asymmetry of the two grades is thus structural: forks are
robust news; canonicity is always provisional, an artifact of the current
resolution that refinement may revoke.

\subsection{The absolute limit, and what records can never establish}

\begin{lemma}[Absolute limit]\label{lem:abslimit}
$\isContr\big(\Fill(\sigma)\big) \to \Canon_V(\sigma)$ for \emph{every}
regime $V$: the center inhabits $\|\Fill\|$; contractibility gives pairwise
identity; Lemma~\ref{lem:idtoregime} converts each path to
$V$-indiscernibility.
\end{lemma}

In the intended discrete-identity models (\S\ref{sec:semantics}),
$\isContr(\Fill)$ is the \emph{strict} Segal condition --- a unique filler
on the nose, the nerve-of-a-category case --- so the comparison with
synthetic $\infty$-category theory \citep{riehlshulman2017} is an analogy
of form together with this strict limit, and we claim nothing more. The
graded judgments $\Canon_V$ are the working notions; the absolute notion is
their common upper bound.

\begin{remark}[Canonicity is never derivable from a finite record]
\label{rem:canonunder}
No \emph{consistent} finite measurement context yields a closed term of
$\Canon_V(\sigma)$. (Consistency matters: an explosive record ---
$\delta : u \apV{V} u$ --- yields everything by
\textsc{$\#$-irrefl} and $\emp$-elimination; honest records are
consistent by Theorem~\ref{thm:recordconsistency}.) Two reasons:
(i) $\MC$-entries are pairwise verdicts
on \emph{named} closed terms, and exhaustiveness of a record is not an
$\MC$-expressible condition; (ii) $\mathsf{hom}$ and $\Tri$ have no
eliminators, so $\Fill(\sigma)$ admits no internal enumeration, and the
$\Pi$ over all fillers is unreachable from closed instances. (A
hypothesized enumeration $e : \Fill(\sigma) \simeq \mathsf{Fin}\,n$ is a
$\Gamma$-variable route --- the \emph{variables hypothesize} side of the
slogan.) The calculus issues fork verdicts from records; canonicity is the
model-side limit notion that experiments \emph{estimate} by
failure-to-refute. \textbf{Forks are verifiable; canonicity is
falsifiable-only.} This asymmetry is not a defect: it is the
falsifiability structure the empirical bridge of \S\ref{sec:empirics}
requires.

The same poverty has a second face worth stating: the theory can say
``these two senses are apart'' but never ``there are exactly three
senses at this horn.'' We hold this to be principled, not accidental.
A cardinality claim is a \emph{different kind of verdict} --- it
asserts exhaustiveness, and exhaustiveness is precisely what no
pairwise instrument delivers; an instrument that could count would be
a new primitive (an enumeration channel), not a refinement of the
ones we have. The calculus accordingly treats counting the way it
treats canonicity: hypothesize it if you wish (an enumeration
$e : \Fill(\sigma) \simeq \mathsf{Fin}\,n$ is a perfectly good
$\Gamma$-variable, and conditional theorems follow), but no record
formed from pairwise verdicts will ever discharge it. A richer
theory with enumeration-grade records is conceivable; it would be
answerable to a different experimental practice.
\end{remark}

\subsection{Cross-regime consistency}\label{ssec:crossregime}

The question Frege's puzzle will pose in \S\ref{sec:sense}: can a
connection discovered by one instrument coexist with a separation recorded
by another? The answer is a theorem, and it depends on nothing but the
order relation between the instruments.

\begin{theorem}[Disagreement forces non-finer connection]
\label{thm:crossregime}
\;
\begin{enumerate}[label=(\alph*)]
\item (Universal.) For all $u, v : S$: if $W \le V$, then
  $(u \simV{W} v) \to (u \apV{V} v) \to \emp$.
\item (Schematic consistency.) For $u, v$ \emph{fresh} constants of a base
  semantic type and $\MC$ exactly
  $\{\gamma : u \simV{W} v,\; \delta : u \apV{V} v\}$:
  $\mathrm{TTS}(\MC)$ is consistent if and only if $W \not\le V$ ---
  that is, $V < W$ strictly, or $V, W$ incomparable.
\end{enumerate}
\end{theorem}
\begin{proof}
(a) \textsc{$\simeq$-mono} at $V' := W$ (side condition exactly
$W \le V$) lifts $u \simV{W} v$ to $u \simV{V} v$; apply
$\mathsf{clash}_V$. (b) Refutation when $W \le V$ is (a) via
($\MC$-ax). Consistency witnesses, with details in
Appendix~\ref{app:soundness}: for $V < W$, the resolution model on two
points at distance $r$ with $\varepsilon_V < r \le \varepsilon_W$; for
incomparable $V, W$, the two-point, two-regime record model
($\apV{V}$ the pair, $\simV{V}$ diagonal; $\simV{W}$ total, $\apV{W}$
empty). In each case the recipe
$\simV{U} := \text{total if } W \le U \text{ else diagonal}$,
$\apV{U} := \text{pair if } U \le V \text{ else empty}$ \emph{defines}
the model at every regime of the poset (the two-point descriptions are
its motivating instances): monotonicity follows from transitivity of
$\le$, and the clash-prone combination would require
$W \le U \le V$, excluded by hypothesis. Freshness matters:
instantiating $u := v$, or adding conflicting entries, refutes by
\textsc{$\#$-irrefl} or id-clash regardless of the regimes.
\end{proof}

At $V = W$ refutation wins --- it is $\mathsf{clash}_V$ itself. The
reading we will lean on: \emph{a configuration in which one instrument
connects what another separates is consistent exactly when the connecting
instrument is neither the separating one nor a refinement of it} ---
strictly coarser, or a different modality altogether. Looking harder
through the same instrument can never identify what it separated.
Example~\ref{ex:miniature} exhibits the consistent configuration live:
$\delta_1 : \varphi_1 \apV{V} \varphi_3$ at resolution $0.60$ and
$\varphi_1 \simV{W} \varphi_3$ at resolution $0.70$ stand together in
one dataset, because $V < W$ --- the connecting instrument is strictly
the coarser one.

\subsection{Readouts and pullback regime structures}\label{ssec:readout}

Instruments compose. One rarely observes fillers directly; one observes
them \emph{through} a readout --- an embedding, a projection, a summary ---
and measures in the image. The calculus internalizes this.

\begin{definition}[Readout; pullback]\label{def:readout}
A \emph{readout} is a function term $\rho : S \to R$ between semantic
types; instrument readouts are function constants of the signature
$\Sigma$. Given a regime $W$ (acting on $R$), the \emph{pullback notation}
\[
u \simV{\rho^* W} v \;:\equiv\; \rho u \simV{W} \rho v
\qquad
u \apV{\rho^* W} v \;:\equiv\; \rho u \apV{W} \rho v
\]
is licensed by unfolding --- \textsc{$\simeq$/$\#$-form} at $W$ on
$\rho u, \rho v$. A subscript $\rho^* W$, in former or $\MC$-entry
position, is an abbreviation, not a new element of the regime poset; no
other non-poset indices are admitted, and models need no clause for them.
\end{definition}

\begin{proposition}[Pullback laws]\label{prop:pullback}
$(\simV{\rho^* W}, \apV{\rho^* W})$ satisfies every single-regime rule of
\S\ref{ssec:apart} at $S$ (each inherited in one line; cotransitivity by
splitting $W$'s cotransitivity at $\rho w$). Derived order implications:
for $W' \le W$, $\;u \simV{\rho^* W'} v \to u \simV{\rho^* W} v$ and
$u \apV{\rho^* W} v \to u \apV{\rho^* W'} v$, each by $W$'s own
monotonicity rule.
\end{proposition}

The consequence for measurement is structural: \emph{a separation observed
in the image is a separation of the originals at the pullback index} ---
the composite instrument is itself an instrument. Its blind spots are
equally structural: if $\rho$ factors through the first projection of
$\Fill$, the pullback index cannot separate fillers differing only in
their warrant. Both facts organize the empirical section
(\S\ref{sec:empirics}).
\section{Sense and reference}\label{sec:sense}

\subsection{The identification}\label{ssec:ident}

\begin{definition}[Reference; sense]\label{def:sense}
Let $\sigma = (f, g)$ be a horn in a semantic type $A$.
\begin{itemize}
\item The \emph{reference} of $\sigma$ is its boundary data:
  $\Refr(\sigma) :\equiv (a, b, c, f, g)$ --- what is already fixed
  and must be respected by any completion.
\item A \emph{sense} at $\sigma$ is an inhabitant
  $\varphi : \Fill_A(\sigma)$ --- a way of completing what is given.
  We write $\Sense_A(\sigma) :\equiv \Fill_A(\sigma)$, the same type
  under its semantic description.
\item Senses $\varphi, \psi$ are \emph{hyperintensionally distinct at
  the regime $V$} when $\varphi \apV{V} \psi$ is inhabited. There is no
  sense-distinctness \emph{simpliciter}: only mathematical identity (the
  absolute limit of Lemma~\ref{lem:abslimit}) and graded distinctness
  relative to an instrument. This relativity is not a concession; it is
  the finding.
\end{itemize}
\end{definition}

The Fregean doctrine falls out as structure --- and the first clause
is a construction, not a typing convention. Let
$\mathsf{Horn}(A)$ be the type of inner $2$-horns
(Definition~\ref{def:filler}: a $\Sigma$-type of boundary data) and
form the \emph{total space of senses}
\[
\Sense(A) \;:\equiv\; \sum_{\sigma : \mathsf{Horn}(A)}
\Fill_A(\sigma),
\qquad
\pi : \Sense(A) \to \mathsf{Horn}(A).
\]
A sense, taken absolutely, is a point of the total space; it
\emph{carries} its reference as its image under $\pi$.

\begin{enumerate}[label=(\alph*)]
\item \emph{Sense determines reference, not conversely.} Determination
  is the projection $\pi$. The converse fails when
  $\Fork_V(\sigma)$ is inhabited: one fibre, several separated
  modes of completion.
\item \emph{No sense without a horn.} Sense is not an intrinsic halo
  attached to an expression; it is relative to a compositional demand.
  The same lexical material in a different horn has a different filler
  space. (This is the type-theoretic rendering of the context
  principle: only in the context of a compositional demand does an
  expression have a sense.)
\item \emph{Where composition is canonical, sense collapses into
  reference.} If $\Canon_V(\sigma)$ holds, all senses at $\sigma$ are
  $V$-indiscernible: one way of arriving, as far as the instrument can
  tell. The sense/reference distinction is literally invisible in a Kan
  world --- which explains why a logic adequate to mathematics, where
  canonicity is the global ambient assumption, kept failing to be a
  logic of natural-language sense.
\end{enumerate}

\subsection{The informativeness of identity, constructively}\label{ssec:inform}

Let $\varphi, \psi : \Sense_A(\sigma)$ be two senses at one horn ---
two modes of completion of one compositional reference --- and let the
record stand at $\delta : \varphi \apV{V} \psi$: some instrument $V$
has separated the two ways of arriving. What can it mean to
\emph{discover} that they are one?

Not literal identity: by id-clash (Corollary~\ref{cor:idclash}), the
recorded separation refutes $\varphi = \psi$ outright --- two
measurably distinct completions are not the same completion, and no
discovery will make them so. What is discovered is
\emph{indiscernibility at another instrument}: a record entry
$\gamma : \varphi \simV{W} \psi$ --- under the new instrument $W$, the
two modes of arriving cannot be told apart. And here the calculus
delivers its sharpest philosophical result. Refutation is universal: by
Theorem~\ref{thm:crossregime}(a), whenever $W \le V$ --- the
discovering instrument is the separating one or a refinement of it ---
the configuration $\{\varphi \apV{V} \psi,\ \varphi \simV{W} \psi\}$
explodes. Consistency in the remaining cases ($W \not\le V$: strictly
coarser, or incomparable --- a different modality altogether) is
established by Theorem~\ref{thm:crossregime}(b) for named base points,
and for filler pairs in the strictly-coarser case by the model of
Proposition~\ref{prop:asym}; the incomparable filler case follows the
same recipe. Looking harder through the same instrument, or any
refinement of it, can never identify what it separated.

This is, we contend, the right anatomy of Frege's datum. The
informativeness of an identification is the fact that it is a
\emph{new instrument's verdict}, consistent with --- not a retraction
of --- the standing separation: the perceptual regime that
distinguishes morning-star-wise from evening-star-wise arriving is not
refuted by the astronomical regime that connects them; the two verdicts
coexist because the instruments are incomparable. Cognitive
significance is the \emph{regime profile} of a filler pair: which
instruments separate them, which connect them, and how those
instruments are ordered. Nothing here invokes belief contexts,
impossible worlds, or syntactic quotation: the hyperintensional grain
is carried by the graded geometry of $\Sense_A(\sigma)$ itself.

\begin{remark}[Scope: from fillers to names]\label{rem:hesperus}
The celestial example involves more than the filler-level machinery
shown so far. ``Hesperus'' and ``Phosphorus'' are distinct
\emph{names}, and Frege's case is two expressions with one referent ---
structurally dual to the one-boundary-many-fillers configuration of
\S\ref{ssec:polysemy}, and a case any full treatment must square with
rigid designation \citep{kripke1980}. \S\ref{ssec:names} delivers the
core: the cross-regime configuration for \emph{name-states} ---
presentation-points named by the signature --- is a theorem-grade
instance of Theorem~\ref{thm:crossregime}, with the informativeness of
identity and the persistence of the modes of presentation as
one-line consequences. What remains for a companion development is the
semantics of the public name \emph{as an expression} --- the
word-to-constant link and its behaviour under predication --- together
with rigidity in full.
\end{remark}

\subsection{Worked example: lexical ambiguity as a fork}
\label{ssec:polysemy}

Let $A$ be a semantic type for an unfolding discourse, and consider
the fragment
\begin{quote}
\emph{``\dots river \dots bank \dots''} \qquad versus \qquad
\emph{``\dots deposit \dots bank \dots''}
\end{quote}
(\emph{Bank} is the classic case of \emph{homonymy} --- two unrelated
lexemes sharing a form; we use it for vividness. Systematic
\emph{polysemy} --- related senses of one lexeme, the copredication
cases --- poses further demands we do not address here; the
type-theoretic treatments of record are
\citep{asher2011, pustejovsky1995, chatzluo2020}, and relating
dot-type coercion to fork structure is future work.) Fix the horn
$\sigma = (f, g)$ in which $f$ is the continuation arriving at the
occurrence of \emph{bank} and $g$ the continuation departing from it
into further discourse that does not itself disambiguate. Then:

\begin{itemize}
\item $\Refr(\sigma)$: the lexical occurrence with its flanking
  commitments --- what any reading must respect.
\item Two fillers $\varphi_{\mathrm{fin}}, \varphi_{\mathrm{riv}} :
  \Fill_A(\sigma)$: each a composite reading \emph{with} a triangle
  witnessing that it actually composes the flanking continuations.
\item A recorded separation $\delta : \varphi_{\mathrm{fin}}
  \apV{V} \varphi_{\mathrm{riv}}$, realized in the model as the metric
  separation of the two contextual completions at the instrument $V$
  (empirically: the two readings' embeddings fall in different
  components at any reasonable resolution; \S\ref{sec:empirics}
  measures exactly this).
\end{itemize}
So $\Fork_V(\sigma)$ is inhabited: ambiguity is a fork, and a
\emph{reading} is a choice of filler. Disambiguating discourse does
not erase the fork --- by Theorem~\ref{thm:persist} a recorded fork
stands at every refinement, and nothing retracts a record --- it
\emph{outruns} it: the disambiguated discourse poses a
\emph{different} horn $\sigma'$, with richer boundary, whose filler
space may be canonical. Ambiguity is resolved not by deleting the
other sense from the language but by moving to a compositional demand
that no longer admits it. This is, recognizably, the dynamic-semantics
picture --- meaning as context-change rather than static content
\citep{kamp1981, heim1982, groenendijkstokhof1991, veltman1996} ---
transposed to a setting where the contexts are geometric.

\begin{remark}[Against sense-as-algorithm, in one paragraph]
Procedural theories individuate senses by identity of algorithms
\citep{moschovakis2006, duzi2010} or programs \citep{bentzen2020};
proof-theoretic semantics individuates the sense of a derivation by
its rule-structure \citep{tranchini2016, ayhan2021}. All locate sense
in the \emph{syntax of arriving}. TTS locates it in the \emph{space of
arrivings}: the individuation is geometric (apartness classes at a
regime), the identity criterion is instrument-relative connection, and
both are answerable to measurement rather than to a chosen formal
presentation.
\end{remark}

\subsection{Names: two routes, one object}\label{ssec:names}

Frege's own case now. Fix a base semantic type $P$ of the signature:
the \emph{presentation space} --- object-directed contextual states
(for a language model: embedded states in which the discourse is
directed at a particular object). Names enter as the signature's
\emph{term constants} valued in $P$ --- the naming discipline of
\S\ref{ssec:semtypes}, doing now for presentation-points what it did
for the miniature's completions: $h, p : P$, the Hesperus-wise and
Phosphorus-wise presentation states. Two scoping commitments, stated
plainly. First, the relation between the public \emph{word}
``Hesperus'' and the constant $h$ is presupposed, not constructed: the
presentational instrument names the clusters it separates, and the
lexicon is assumed to label them; the semantics of the name as an
expression is the companion's subject
(Remark~\ref{rem:hesperus}). Second, we do not call $h$ and $p$
\emph{senses}: by the context principle of \S\ref{ssec:ident}(b) senses live at
horns, and $h, p$ are points. Cognitive significance for names will be
the regime profile of a \emph{point pair} --- the point-level analogue
of \S\ref{ssec:inform}'s filler-pair notion; the horn-level story (name-states as
boundary vertices of predicational demands) belongs to the deferred
predication dynamics.

Two regimes carry the case, both ordinary poset regimes. $V$ is a
\emph{presentational} instrument --- it separates
presentation-clusters, as the perceptual and temporal circumstances of
observation do: evening-presentations from morning-presentations. Its
verdict stands in the record: $\delta : h \apV{V} p$ --- in any
realizing model, the two name-states lie in distinct $V$-components
(presentation-basins, when those components are canonical;
basin-hood itself, per Definition~\ref{def:basin}, is
falsifiable-only). $W$ is a different modality --- positional,
calculational: the astronomical instrument --- and the
\emph{discovery} is its record entry: $\gamma : h \simV{W} p$.

\begin{corollary}[The names configuration]\label{cor:names}
For $h, p$ fresh constants of the base semantic type $P$ ($h \neq p$
as symbols, occurring in no other $\Sigma$-declaration and no other
record entry) and $\MC$ exactly
$\{\delta : h \apV{V} p,\; \gamma : h \simV{W} p\}$:
$\mathrm{TTS}(\MC)$ is consistent if and only if $W \not\le V$. This
is Theorem~\ref{thm:crossregime}(b) verbatim --- the names
configuration is its proven scope, fresh constants of a base semantic
type.
\end{corollary}

Where is the object? Designate an \emph{objectual} regime $W^{*}$ with
$W \le W^{*}$ --- the discovering instrument, or a coarsening of it;
which regime counts as objectual is an application-level designation,
like the choice of a metric, and the constraint $W \le W^{*}$ is
itself meaningful: \emph{objecthood is what survives coarsening from
the discovery}. The \emph{reference} of a name-state is then its
$\simV{W^{*}}$-class, and co-reference follows from the discovery by
\textsc{$\simeq$-mono}: $h \simV{W^{*}} p$ --- two routes, one
object. (Reference so defined is name-level and regime-relative;
$\Refr(\sigma)$ of Definition~\ref{def:sense} is demand-level. Frege
himself bifurcates reference by category --- the \emph{Bedeutung} of a
sentence is not that of a name --- and so do we: the two notions meet
when name-states occur as boundary vertices of a horn. The relativity
to $W^{*}$ is not a concession but the position of
Definition~\ref{def:sense}, item three, applied to reference.)

(A seam we mark rather than smooth: the ambiguity case of
\S\ref{ssec:polysemy} realizes the \emph{sense-as-filler} account
directly; the present section realizes the calculus of
\emph{instrument-relative identification}. What unites them is
Theorem~\ref{thm:crossregime} --- one consistency law, two typing
instantiations --- not a single reduction of both phenomena to
fillers.) The Fregean data now fall out as one-line derivations.

\begin{proposition}[Informativeness as record-dependence]
\label{prop:nameinfo}
(i) $\vdash_{\emptyset} \mathsf{r}_W(h) : h \simV{W} h$ for every
regime $W$: trivial identities are derivable from the \emph{empty}
record. (ii) For distinct fresh name constants $h, p : P$ there is no
closed term $\vdash_{\emptyset} t : h \simV{W} p$, nor of $h =_P p$.
\end{proposition}
\begin{proof}
(i) is \textsc{$\simeq$-refl} at the constant ($h : P$ by
Definition~\ref{def:semsig}, ``names for points''). (ii) is semantic
--- erasure cannot show it, since $(\,\simV{}\,)^{\circ} = \mathbf{1}$
--- by the two-point model of Appendix~\ref{app:soundness}:
interpret $P$ as $\{x, y\}$ with $\Sem{h} = x$, $\Sem{p} = y$, every
$\simV{U}$ the diagonal and every $\apV{U}$ empty, the rest of the
signature as in the base-consistency lemma; the model realizes the
empty record and empties both types; conclude by soundness.
\end{proof}

``Hesperus is Hesperus'' costs nothing and says nothing: it is
reflexivity, available at every instrument, from no record.
``Hesperus is Phosphorus'' is underivable from the empty record ---
\emph{no record, no identification} --- and enters only by citing an
instrument: ($\MC$-ax) on $\gamma$. The informativeness of an
identity statement is, exactly, its record-dependence.

\begin{proposition}[Segregation and transfer]\label{prop:namestransfer}
Let $e$ range over indiscernibility proofs. From the record of
Corollary~\ref{cor:names}:
\begin{enumerate}[label=(\alph*)]
\item (Segregation at the presentational regime; needs only
  $\delta$.) For every $q : P$:
  $(h \simV{V} q) \to (q \apV{V} p)$, by \textsc{$\#$-ext} at
  $(h, q, p)$: $\lambda e.\, \mathsf{ext}(e, \delta)$. Whatever the
  presentational instrument groups with $h$, it must separate from
  $p$ --- and derivability is monotone under record extension, so
  this holds whatever a \emph{consistent} record additionally
  contains (for the discovery pair, consistency is exactly
  $W \not\le V$).
\item (Transfer at the discovering regime.) For every $q : P$:
  $(h \simV{W} q) \to (p \simV{W} q)$, by
  $\lambda e.\, (\gamma^{\dagger} \cdot e)$; and
  $(h \apV{W} q) \to (p \apV{W} q)$, by
  $\lambda \delta'.\, \mathsf{ext}(\gamma, \delta')$ --- and
  symmetrically in each case, interchanging $\gamma$ and
  $\gamma^{\dagger}$.
  At the instrument that made the discovery, the two names'
  verdicts flow freely across the identification.
\end{enumerate}
\end{proposition}

This pair is the anatomy of what discovery does and does not do. At
$W$, the identification licenses full information flow. At $V$,
nothing transfers --- no rule carries $\simV{W}$ down to the finer
$V$ --- and by (a) the presentational neighbourhoods of the two names
remain provably segregated. Learning that Hesperus is Phosphorus does
not, and cannot, merge the presentations: the persistence of
cognitive significance after discovery is a theorem, and it needs
only $\delta$. (A record claiming both $V$-groupings --- some closed
$q$ with $h \simV{V} q$ and $q \simV{V} p$ on record --- refutes
itself at $V$ by (a) and the clash lemma; by the contrapositive of
Theorem~\ref{thm:recordconsistency}, no honest instrument run
produces it.)

\begin{remark}[Rigidity, conditionally]\label{rem:rigidity}
Within any model, a signature constant's denotation is
environment-independent, while an open descriptive term's varies with
its environment --- a proof-theoretic shadow of the rigid/non-rigid
contrast \citep{kripke1980}. Rigidity proper --- stability across
counterfactual evaluation --- requires modal or diachronic machinery
TTS deliberately lacks (record-counterfactuals are not internalizable,
Remark~\ref{rem:provsurplus}), and is the companion's subject.
\end{remark}

\section{Semantics and metatheory}\label{sec:semantics}

Sections~\ref{sec:calculus}--\ref{sec:sense} are a calculus and its
philosophical reading; we owe the mathematics. This section defines the
models, states soundness, and proves the metatheorems that make the
measurement-context architecture more than rhetoric: the calculus is
conservative over its base, creates no forks of its own, and derives
nothing the record does not force.

\subsection{Models}\label{ssec:models}

\begin{definition}[Apartness space]\label{def:apspace}
An \emph{apartness space} $(X, \apart)$ is a set with an irreflexive,
symmetric, cotransitive relation \citep{bishopbridges1985,
troelstravandalen1988}. Every metric space yields, at each resolution
$\varepsilon$, the apartness ``lying in different
$\varepsilon$-components'', where the $\varepsilon$-chain relation (steps
of length $\le \varepsilon$) supplies the matching indiscernibility.
\end{definition}

\begin{definition}[Regime-indexed simplicial apartness space]\label{def:sas}
A \emph{regime-indexed simplicial apartness space} is a simplicial object
$X$ in apartness spaces such that:
\begin{enumerate}[label=(M\arabic*)]
\item each level $X_n$ carries, for each regime $V$, a pair
  $(\simV{V}, \apV{V})$ of h-proposition-valued relations. The
  denotation of a semantic type is a subspace of a level or a
  $\Sigma$-assembly of such (hom and $\Tri$ over $\Sigma$-assemblies
  interpreted fibrewise); each denotation carries the \emph{induced}
  structure, with components computed \emph{intrinsically}: two readings
  count as connected only if deformable into one another through
  readings, not through arbitrary ambient points --- which is what makes
  a fork verdict about the filler space rather than the embedding cloud.
  On $\Sigma$-assemblies the induced structure is the total-space one:
  in metric models, a step between pairs is measured as the maximum of
  the componentwise steps, so that pairs sharing a component are exactly
  as far apart as their other components. The rules of
  \S\ref{ssec:apart} hold per denotation;
\item faces are strongly extensional (metric models: nonexpansive);
  degeneracies preserve and reflect both relations at every regime
  (metric models: isometric) --- read at the level structures;
\item \textbf{(discrete identity)} the ambient identity type is
  interpreted discretely: only reflexivity. Coarse path-object
  interpretations of identity are \emph{not} members of the model class:
  with non-trivial identity, unrestricted path induction reaches
  hom- and $\Tri$-dependent families
  ($\transp : (a =_S b) \to \hom{S}{a}{c} \to \hom{S}{b}{c}$)
  which no apartness axiom governs. Discreteness makes every instance of
  path induction a reflexivity-transport, trivially sound --- and costs
  nothing: every model used in this paper is discrete in this sense. All
  non-trivial sameness is carried by the regime-indexed
  $\simV{V}$, which is precisely the design.
\end{enumerate}
A model \emph{realizes} $\MC$ when it validates every recorded verdict.
\end{definition}

\begin{remark}[What TTS is, positioned plainly]\label{rem:position}
Under (M3) the interesting geometry does not live in the identity
types: in the intended models, identity is discrete and all
non-trivial sameness is carried by the regime-indexed $\simV{V}$. So
TTS is not a repair of HoTT's treatment of semantic composition ---
it \emph{brackets} homotopical identity and adds an observational
layer: intensional MLTT with primitive directed types and graded,
record-fed observational indiscernibility and apartness. We take this
to be the contribution, not a concession: the layer is where
measurement lives, and the metatheorems of
\S\ref{ssec:metatheorems} are about the layer. (The kinship is less
with simplicial HoTT than with observational equality
\citep{altenkirch2007}, graded by instruments and grounded in
records.)
\end{remark}

\begin{example}[The resolution model]\label{ex:resolution}
Let $X_0$ be a set of contextual states embedded in a metric space (for
the intended case: embedded states of a generative model under
$d_{\cos}$), with levels assembled from observed or realizable
transitions and compositions. For each resolution $\varepsilon$:
$u \simV{\varepsilon} v$ iff $u, v$ are joined by a chain of steps of
length $\le \varepsilon$ within the relevant denotation;
$u \apV{\varepsilon} v$ iff they lie in different $\varepsilon$-components.
All rules of \S\ref{ssec:apart} hold (cotransitivity: a third point lies
in a different component from at least one side; monotonicity:
$\varepsilon'$-chains are $\varepsilon$-chains, so components refine);
identity is discrete. The estimator of \S\ref{sec:empirics} --- connected
components of the $\varepsilon$-neighbourhood graph on a sample ---
computes the apartness classes of this model on the sample: the experiment
is the semantics, evaluated on data.
\end{example}

\subsection{Soundness and the honesty of records}

\begin{theorem}[Soundness, environment form]\label{thm:soundness}
If $\Gamma \vdash_{\MC} t : A$, then in every regime-indexed simplicial
apartness space realizing $\MC$, $\Sem{t}$ is a section of $\Sem{A}$ over
$\Sem{\Gamma}$.
\end{theorem}
\begin{proof}[Proof (cases in Appendix~\ref{app:soundness})]
Induction on derivations. The MLTT cases are standard; path induction is
reflexivity-transport by (M3), hence sound for \emph{arbitrary} motives ---
including hom-, $\Tri$-, $\simeq$- and $\#$-dependent families. The
directed rules are immediate from the simplicial structure; the
relational rules hold per denotation by (M1); ($\MC$-ax) by realization.
Substitution coherence: the stated convention plus the standard
comprehension-category treatment of the MLTT fragment; the full
categorical elaboration is routine under (M3) and deferred.
\end{proof}

\begin{theorem}[Consistency of honest records]\label{thm:recordconsistency}
Let $P$ be a finite sample of a metric space, with the induced metric, and
finitely many resolutions. The record $\MC_P$ of $\varepsilon$-component
verdicts \emph{on the sample} is realized by the resolution model on $P$
--- the record \emph{is} that model's component structure --- so
$\mathrm{TTS}(\MC_P)$ is consistent.
\end{theorem}
\begin{proof}
Equip $P$ with discrete simplicial levels (the sampled completions as
the relevant fibre), chain/component relations at every resolution, and
discrete identity: a model per Definition~\ref{def:sas} whose verdicts
at the swept resolutions are by construction the recorded ones; conclude
by Theorem~\ref{thm:soundness}.
\end{proof}

Two glosses, both load-bearing. First, the record is a record \emph{about
the sample}: sample-apart points may be connected in the ambient space
through unsampled territory; the instrument's verdicts quantify over what
it measured, which is what the resolution model on $P$ interprets. Second,
the guarantee is a property of \emph{honest measurement}, not of the
record format: a stipulated record can explode ($\delta : u \apV{V} u$
refutes itself by \textsc{$\#$-irrefl}), and an honest sweep never records
a Theorem~\ref{thm:crossregime}-refutable pair, because its components are
consistent across its own resolutions. The calculus does not certify the
record; the metatheory certifies the instrument.

\subsection{Conservativity and provenance}\label{ssec:metatheorems}

\begin{theorem}[Conservativity]\label{thm:conserv}
For $\Gamma, A$ in the $\simeq/\#$-free fragment (hom and $\Tri$ allowed):
if $\Gamma \vdash_{\emptyset} t : A$ then $\Gamma \vdash t^{\circ} : A$ is
already derivable in MLTT${}+{}$hom/$\Tri$. (Erasure translation;
Appendix~\ref{app:soundness}.)
\end{theorem}

\begin{corollary}[No fork ex nihilo]\label{cor:noforkexnihilo}
There is no closed term $\vdash_{\emptyset} t : \Fork_V(\sigma)$, nor of
the judgment of record $\|\Fork_V(\sigma)\|$, for any horn and regime.
\end{corollary}

\begin{corollary}[No refutation ex nihilo]\label{cor:norefutation}
Dually, there is no closed term
$\vdash_{\emptyset} t : (u \simV{V} v) \to \emp$: the empty record
refutes no indiscernibility either. (Erasure sends the type to
$\mathbf{1} \to \emp$; apply to $\star$.)
\end{corollary}

\begin{theorem}[Provenance --- closed judgments]\label{thm:provenance}
If $\vdash_{\MC} t : \Fork_V(\sigma)$ (empty $\Gamma$, $\sigma$ closed;
likewise for $\|\Fork\|$), then $\Fork_V(\sigma)$ holds in every model
realizing $\MC$. Hence if \emph{some} model realizes $\MC$ without forking
$\sigma$ at $V$, no derivation exists: forks are never created by the
calculus, only transported from the record.
\end{theorem}

\begin{remark}[What provenance adds]\label{rem:provsurplus}
By the deduction theorem, $\vdash_{\MC}$-derivability is equivalent to
$\vdash_{\emptyset}$-derivability from the record's propositions as
hypotheses; the theorem's content beyond conditionalization is threefold.
(i) With Theorem~\ref{thm:recordconsistency}, the conditional is
\emph{non-vacuous} for honest records: a realizing model exists, so the
verdict concerns an actually realizable configuration --- whereas
conditionalization holds vacuously even for explosive records. (ii) The
empty-record case is secured \emph{syntactically} by erasure, independent
of any model class. (iii) Record constants are not abstractable, so
record-counterfactuals are not internalizable: categoricity relative to
the actual record is a property of the theory, not of a sequent. We also
state the downgrade plainly: provenance is model-theoretic; no syntactic
witness-extraction is claimed, since normalization for the axiomatic
extension is not claimed.
\end{remark}

\subsection{Basins, strata, and the third grade}\label{ssec:basins}

\begin{definition}[Points; $V$-components]\label{def:points}
The \emph{points} of a semantic type $S$ are the elements of $\Sem{S}$ in
a fixed model (syntactic reading: closed terms, identified when
$\vdash_{\MC} u \simV{V} v$ is inhabited --- an equivalence by
refl/sym/trans). The \emph{$V$-component} of a point is its
$\simV{V}$-class; components partition the points.
\end{definition}

\begin{definition}[Basin at a regime]\label{def:basin}
A $V$-component $C$ of $S$ is a \emph{basin at $V$} if every inner
$2$-horn whose three boundary vertices lie in $C$ satisfies
$\Canon_V(\sigma)$ --- ambient canonicity, full $\Fill(\sigma)$, no
relativized notion. No maximality is invoked: a basin, when it exists, is
the unique component of its points; whether a given component \emph{is} a
basin is a model-side property --- like canonicity itself
(Remark~\ref{rem:canonunder}), refutable by finite records, never
certified by them.
\end{definition}

\begin{corollary}[Basins exclude forks]\label{cor:basinfork}
If $C$ is a basin at $V$ and $\sigma$ has all three vertices in $C$, then
$\Fork_V(\sigma) \to \emp$, by Exclusion at the same ambient filler space.
\end{corollary}

\begin{remark}[Stratification, exactly]\label{rem:strat}
The contrapositive is precise: \emph{no basin contains all three vertices
of a forked horn}. In particular, if a forked horn's vertices are
co-located in a single $V$-component, that component is not a basin ---
the ambiguous occurrence is a singular point. When the vertices span
components, the forked horn forces the ambiguous middle vertex \emph{out
of the flanking contexts' basin}; whether its own component is a basin is
a further check, not a consequence. This is the formal counterpart of the
stratified picture of \citep{robinson2024structure, li2025stratified}:
basins are the strata where travel is forced; fork loci are where strata
meet.
\end{remark}

\begin{definition}[The third grade: locked]\label{def:fer}
A basin $C$ at $V$ is \emph{locked} if every escape horn from $C$ --- a
horn with some boundary vertex in $C$ and some outside --- has empty
filler space, $\Fill(\sigma) \to \emp$. (Emptiness is regime-independent;
the regime enters solely through basin-hood.) A trajectory confined to a
locked basin exhibits \emph{over-coherence}: maximal local fluency, zero
semantic decision --- the degenerate-fluency phase documented in
repetition collapse and low-temperature locking \citep{arnold2024phase,
nakaishi2026critical}. The criticality finding of
\citep{nakaishi2026critical} --- natural-language statistics arise at the
phase boundary --- receives a structural gloss: meaning requires residence
near the fork locus; both the locked basin and the fully disconnected
regime are semantically dead. We offer this as an interpretation, not a
theorem.
\end{definition}

The three conditions of \S\ref{sec:intro} are now formal --- with
their typing stated honestly. \textbf{Canonical} ($\Canon_V$) and
\textbf{forked} ($\Fork_V$) are the two observational grades \emph{of
a horn}; they do not exhaust verdicts (undetermined is the generic
state of an unexamined horn). \textbf{Locked} is not a third parallel
grade: it is a basin-level construction \emph{from} canonicity ---
a dynamical pathology of a region, not a verdict on a demand. The
triple is a phenomenology of travel; the trichotomy-sounding talk of
\S\ref{sec:intro} should be heard that way. One asymmetry organizes
everything: forks persist under refinement and are verifiable from
finite records; canonicity and basin-hood are falsifiable-only;
locked-ness adds an emptiness condition that no record can certify
either. The logic's positive verdicts are exactly the ones
measurement can deliver.
\section{The empirical bridge}\label{sec:empirics}

The resolution model was chosen for more than soundness: every
constituent of it is computable from samples. A horn is a prompt with
flanking commitments; fillers are sampled continuations; a regime is an
embedding model with a resolution; apartness classes are components of
the $\varepsilon$-neighbourhood graph; and the readout through which
fillers are observed is a $\Sigma$-constant whose pullback
(Definition~\ref{def:readout}) is itself a regime structure. The
dictionary is Table~\ref{tab:correspondence}; the reader is invited to
keep it open --- every row is exercised below. We state the
correspondences as a falsifiable protocol, with the limits of current
evidence marked --- including, prominently, the scope limits of our own
pilot.

\begin{table}[t]
\centering\small
\begin{tabular}{@{}p{0.36\linewidth}p{0.40\linewidth}l@{}}
\toprule
\textbf{Experimental operation} & \textbf{Formal object} & \textbf{Where} \\
\midrule
prompt with flanking commitments & inner $2$-horn $\sigma$ (cone if the
  suffix is free; genuine horn in the FIM test) & Def.~\ref{def:filler} \\
sampled continuation & readout image of a filler & \S\ref{ssec:tri},
  Def.~\ref{def:readout} \\
embedding model $+$ resolution $\varepsilon$ & regime
  $W = (d_{\cos}, \varepsilon)$ & Def.~\ref{def:regime} \\
embedding map (first sentence) & readout $\rho$; pullback index
  $\rho^* W$ & Def.~\ref{def:readout} \\
$\varepsilon$-graph components on the sample & apartness classes of the
  resolution model on $P$ & Ex.~\ref{ex:resolution} \\
the sweep's component verdicts & measurement context $\MC_P$ &
  Def.~\ref{def:mctx}, Thm.~\ref{thm:recordconsistency} \\
$\ge 2$ stable sizable components & $\vdash_{\MC_P}$ apartness verdict at
  $\rho^* W$ & Def.~\ref{def:fork}, scope note \\
one component throughout & failure to refute $\Canon$ --- estimator
  verdict only & Rem.~\ref{rem:canonunder} \\
control conditions, their means & putative basins, their
  representatives & Def.~\ref{def:basin} \\
incoherent generator (GPT-2 null) & undetermined --- no verdict either
  way & \S\ref{sec:empirics} \\
\bottomrule
\end{tabular}
\caption{The dictionary between TTS and the fork-test protocol. By
Example~\ref{ex:resolution}, the left column is not an analogy for the
right: it computes it, on a finite sample.}
\label{tab:correspondence}
\end{table}

\paragraph{Manifold-hypothesis violation $\leadsto$ candidate fork
loci.} \citep{robinson2025manifold} rejects, by a fiber-bundle null
test on radial volume growth, the hypothesis that token subspaces are
manifolds; singular neighbourhoods are concentrated on identifiable
tokens, and prompts containing them yield measurably less stable
continuations (their Theorem~2 propagates input singularities to
outputs). In TTS terms: loci where local dimension is undefined are
exactly where no single basin chart covers the boundary data ---
candidate inhabitants of the fork locus. The inference is
directional: non-manifoldness certifies that canonical-everywhere
fails; it does not by itself produce a fork witness.

\paragraph{Forking tokens $\leadsto$ apartness verdicts.}
\citep{bigelow2025forking} resamples generation token-by-token and
detects positions where the downstream outcome distribution splits
into discrete alternatives. Read in TTS: the resampled continuations
are observed, through a readout, as draws from a space of completions;
the detected outcome-classes are apartness components at the pullback
index; a forking token is a position at which the protocol has
\emph{recorded} the components of a separation --- two concrete
completions apart at the declared resolution.

\paragraph{Scope: cones and horns.} A scoping point we state
plainly, because the formal subject of $\Fork_V$ is the inner
$2$-horn. Free-completion protocols --- Bigelow et al.'s and our first
pilot below --- fix only the \emph{arriving} commitment: the measured
object is a \emph{cone} (context, awaiting continuation), not a
$2$-horn (both flanking commitments held fixed). The verdicts such
protocols deliver are apartness verdicts on sampled completions at the
pullback index $\rho^* W$ --- fully inside the theory by
Proposition~\ref{prop:pullback} --- but they instantiate no
$\Fork_V(\sigma)$. The genuine $2$-horn test is, in machine-learning
terms, a \textbf{fill-in-the-middle (FIM) evaluation}
\citep{donahue2020, bavarian2022}: hold prefix \emph{and} suffix
fixed, sample the bridge --- a direct sample of $\Fill(\sigma)$'s
composites, with the model's fluent generation under both flanking
commitments as the operational shadow of the warrant.

\paragraph{The two FIM instruments.} A language model never generates
backwards; fill-in-the-middle works by moving the future into the
prompt, and there are two ways to do it. \emph{Native} FIM is trained
in \citep{bavarian2022}: pretraining documents are cut into (prefix,
middle, suffix) and re-ordered as prefix--suffix--middle behind
sentinel markers, so that at inference the model receives both flanks
in its context window and writes the middle left-to-right,
conditioned on what must come after --- the mechanism behind code
editors' insert mode, and exposed for our model as the completion
API's suffix parameter. What we will call \emph{prompted} FIM --- the
label is ours; the practice, infilling via instruction, is standard
and descends from \citet{donahue2020} --- carries both flanks inside
an ordinary prompt (``story before the gap \dots\ story after the
gap \dots\ the missing middle:'') and asks for the bridge. The two
instruments sample \emph{different conditional distributions}: native
FIM approximates the document model's
$P(\text{middle} \mid \text{prefix}, \text{suffix})$; prompted FIM
samples an instruction-conditioned variant in which the demand that
the whole ``read as one coherent whole'' is part of the conditioning.
Neither is the other's refinement. They are, in the paper's own
vocabulary, incomparable regimes --- and we run both. One blind spot is structural either
way: our $\rho$ factors through the generated text, so the measured
index cannot separate fillers differing only in their warrant ---
precisely the configuration of Proposition~\ref{prop:asym}'s
countermodel.

\paragraph{The fork test.} The protocol, then: given a boundary, sample
$N$ continuations, embed completions ($\rho$), and compute the
components of the $\varepsilon$-neighbourhood graph across a sweep ---
by Example~\ref{ex:resolution}, these are the apartness classes of the
resolution model at the pullback index, evaluated on the sample. The
verdict criterion must respect two facts. First, the asymmetry of the
grades (Theorem~\ref{thm:persist}, Proposition~\ref{prop:asym}):
separations persist downward through the sweep; single-component
verdicts are always provisional. Second, on a finite sample every
point is eventually its own component as $\varepsilon \to 0$, so
``stable under refinement'' needs a rejection region: we adopt a
persistence-interval criterion --- a fork verdict requires $\ge 2$
components, each of size $\ge m$, across all $\varepsilon$ in a window
of length $\ge L$, with $(m, L)$ calibrated against a resampled
unimodal null. The component-count clause is licensed by monotonicity;
the size floor and null calibration are an independent statistical
stipulation whose job is to reject sampling dust --- they are not
consequences of the calculus, and we do not present them as such.
``Canonical'' verdicts, by Remark~\ref{rem:canonunder}, are estimator
verdicts --- failure to refute --- never calculus theorems. The sweep's
component verdicts --- one apartness entry per separated pair per
resolution, one indiscernibility entry per within-component pair ---
\emph{constitute} the measurement context $\MC_P$ of
Theorem~\ref{thm:recordconsistency}, and the protocol's fork verdict is
the closed judgment $\vdash_{\MC_P} \|\cdot\|$-apartness at the pullback
index: the experiment writes the record; the calculus derives what it
forces.

\paragraph{Pilot study.} We ran the test on the ambiguity case of
\S\ref{ssec:polysemy} --- in cone configuration, per the scoping note
above. (The four completions of Example~\ref{ex:miniature} are drawn
from this run; the reader has been looking at its data since
\S\ref{sec:calculus}.) Protocol: three boundaries --- an ambiguous
occurrence of \emph{bank} (AMB), a financially disambiguated control
(FIN), a riparian control (RIV) --- each sampled $N{=}300$ times from
a public completion model (\texttt{gpt-3.5-turbo-instruct}) at
temperature $1.0$ without truncation (unbiased draws from the model's
filler distribution); completions embedded
(\texttt{text-embedding-3-small}, first sentence --- the
sense-committing region); per draw, a scalar \emph{sense margin}
(cosine to a financial anchor minus cosine to a riparian anchor ---
itself a second readout in the sense of Definition~\ref{def:readout},
measuring at its own pullback index); mixture model selection (BIC,
$k{=}1$ vs $k{=}2$) on margins, and the component sweep above on full
embeddings. Throughout the findings, ``basin'' is used in its estimator
sense: the control conditions anchor two putative basins (basin-hood
being falsifiable-only, their unimodal behaviour is failure-to-refute),
with condition means as representatives.
Three findings (Figure~\ref{fig:pilot}):
\begin{enumerate}[label=(\alph*)]
\item \emph{The fork signature is detectable, after calibration.} The ambiguous
  boundary's margin distribution is decisively bimodal
  ($\Delta\mathrm{BIC}(2{-}1) = -39.3$), with its minor mode at
  $-0.196$, within a fraction of a standard deviation of the riparian
  control's mean ($-0.218$), and its major mode on the financial
  side: the two populations of fillers align with the two reference
  basins. The riparian control is unimodal
  ($\Delta\mathrm{BIC} = +15.4$).
\item \emph{Boundary balance is instrument-critical, as the theory
  predicts.} An earlier, uncalibrated ambiguous boundary drew
  $169/31$ financial-vs-riparian fillers and registered only
  marginally; a calibration pass over candidate boundaries (selecting
  the one whose $n{=}50$ pre-sample split most evenly) preceded the
  main run. A fork judgment requires only two separated fillers, not
  balance; but \emph{detection power} at fixed $N$ degrades with
  imbalance.
\item \emph{One-dimensional bimodality is not the fork criterion ---
  basin alignment is.} The financial control is itself
  margin-bimodal ($\Delta\mathrm{BIC} = -26.9$): inspection shows the
  split separates banking-dense completions from same-scenario
  atmospherics --- two regions \emph{within one basin}, both on the
  financial side of the inter-basin midpoint. This is exactly what
  the resolution model requires: apartness is separation in the full
  space (different components), not multimodality along a contrast
  axis. The fork verdict for AMB stands because its modes straddle
  the basins; the FIN split does not qualify because its modes do
  not. (In the calculus's own terms: the margin readout and the
  full-embedding readout are \emph{incomparable instruments}, and their
  disagreement on FIN is licensed by
  Theorem~\ref{thm:crossregime} --- the fork criterion is the
  full-space instrument's.)
\end{enumerate}
A baseline run with GPT-2-small under the same protocol detected no
structure in any condition (sense ordering visible only in condition
means). The null is not a canonicity verdict: it is
\emph{undeterminedness relative to an inadequate instrument} --- the
generator must be coherent enough to commit to a reading within the
sampled window for the readout to measure anything.

\paragraph{Data availability.} Nothing here is offered on trust.
Every sampled text in every run --- $3{,}630$ records across the three
cone pilots, the boundary calibration, both FIM instruments, the failed
\texttt{gpt-4o-mini} calibration, and the $2$-horn experiment --- is
logged verbatim with its full prompt context, and the supplementary
package (logs, prompts and templates, embeddings, results, and all
scripts, with a file-by-file manifest) accompanies the
paper.\footnote{Draft supplement:
\url{https://cassie.tanazur.org/drafts/tts-supplement-v0.4.4.zip}.}
Appendix~\ref{app:logs} reproduces, for each key condition, the exact
prompt and a seeded random sample of the logged outputs, unabridged,
indexed into the package --- including the mode-collapsed
\texttt{gpt-4o-mini} bridges behind the instrument finding.

\begin{figure}[t]
\centering
\includegraphics[width=0.99\linewidth]{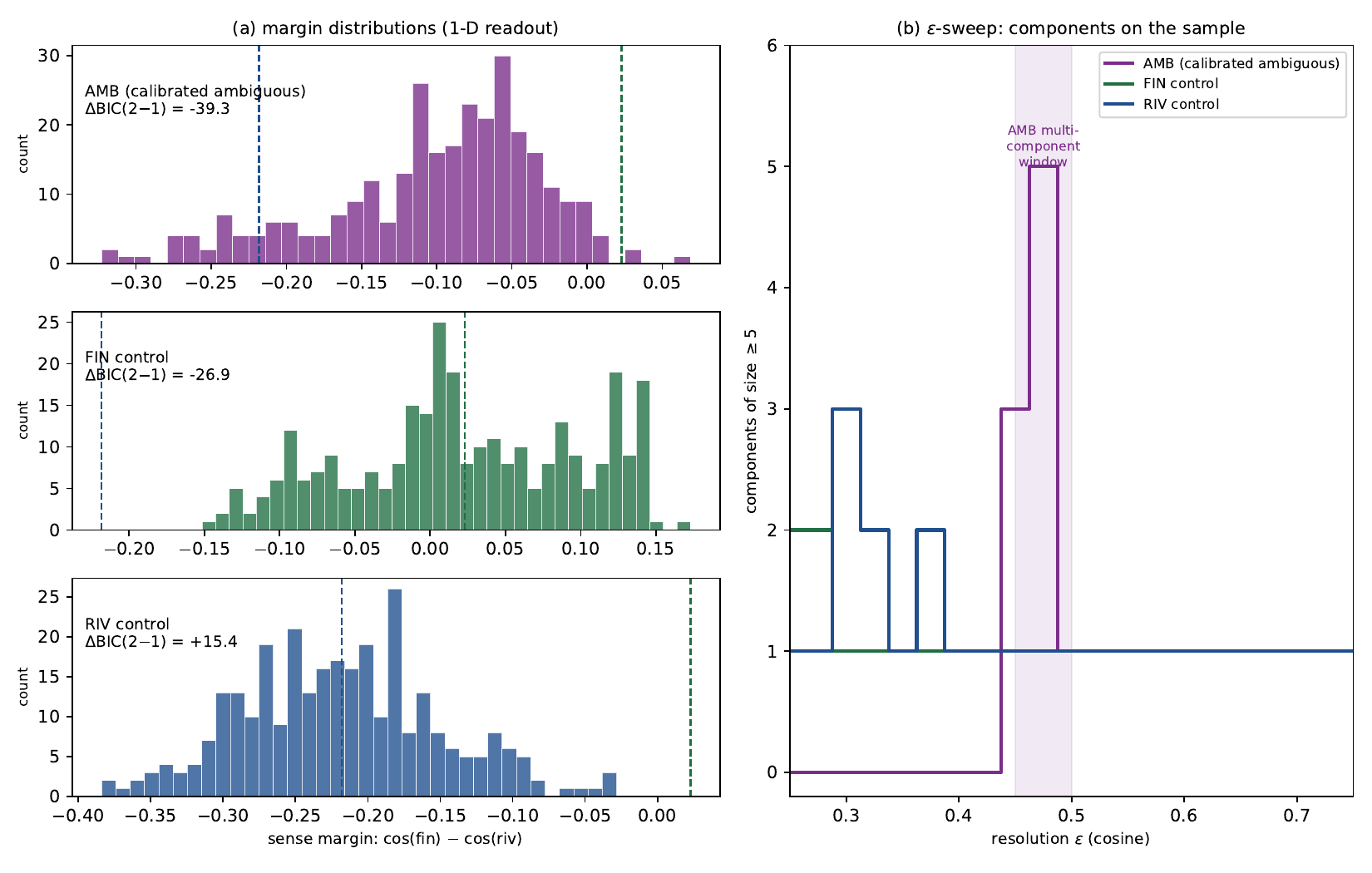}
\caption{The pilot, in both readouts ($300$ sampled completions per
boundary; cone configuration; apartness verdicts at the pullback
index). \emph{(a)} Sense-margin distributions; dashed lines are the
control-condition means (basin representatives in the estimator sense).
The ambiguous boundary spans both basins with modes aligned to each ---
the fork signature; each control remains within its own basin.
\emph{(b)} The $\varepsilon$-sweep: sizable components
($\ge 5$ members) of the sample's $\varepsilon$-neighbourhood graph.
The ambiguous boundary alone holds multiple sizable components in the
shaded window; the controls' fine-$\varepsilon$ fragmentation below
$0.40$ is the sampling dust the persistence criterion's $(m, L)$
calibration exists to reject.}
\label{fig:pilot}
\end{figure}

\paragraph{The 2-horn experiment: infilling.} We then ran the genuine
$2$-horn test as a FIM evaluation. Design: \emph{one} ambiguous prefix
--- the same calibrated boundary (``We met them near the bank, just
before sunset.'') --- and \emph{three} suffixes: neutral (NEU, ``By
the time we said our goodbyes, the stars were already out.''),
financially disambiguating (FIN, a manager unlocking doors for the
last customers), and riparian (RIV, a boat carried downstream). For
each horn, $n = 200$ bridges sampled at temperature $1.0$; same
instrument as the pilot. Two instrument notes, both of independent ML
interest. First, an instruction-tuned chat model
(\texttt{gpt-4o-mini}) failed calibration: its bridge distribution is
mode-collapsed (near-identical sociable boilerplate) and it
under-weights the suffix constraint --- so the experiment uses
prompted FIM on the completion model (\texttt{gpt-3.5-turbo-instruct}),
which retains sampling diversity and respects both flanks. The
instrument requirement is itself a finding: alignment tuning narrows
the bridge distribution toward mode collapse. (Precisely put: no
record certifies canonicity (Remark~\ref{rem:canonunder}); what
alignment manufactures is ever-stronger \emph{failure-to-refute} ---
observed unimodality that the fork test cannot break.)
Second, all runs and boundaries are disclosed; the full-scale study
will pre-register boundaries, sample sizes and the $(m, L)$ criterion.

Results (Figure~\ref{fig:horn}). The neutral-suffix horn is
decisively bimodal on the margin readout
($\Delta\mathrm{BIC}(2{-}1) = -17.8$): a riverside-committed mode at
$-0.209$, aligned with the riparian condition's mean ($-0.223$), and
a second mode at $-0.104$ consisting largely of bridges that
\emph{defer} the reading --- sociable narrative that completes the
passage without committing. Deferral is a FIM-specific phenomenon
the cone test cannot exhibit: a filler whose own reading remains
unresolved, licensed exactly when the right context demands no
commitment. We state the verdict with the paper's own discipline.
By the standard of \S6's protocol --- persistent multi-component
structure in the full embedding space --- this sample does
\emph{not} certify an inhabitant of $\Fork$: the bridges chain
through their shared narrative register, and the
$\varepsilon$-component criterion does not trigger at $n = 200$.
What the experiment establishes is the \emph{fork signature on the
margin readout} (itself a pullback instrument), with modes aligned
to distinct basins --- branching evidence, by the paper's own
finding~(c) standard, short of a certified fork. The certified-fork
experiment needs either larger $n$ or a readout that suppresses
register and isolates reading. The decisive evidence is the
\textbf{suffix effect}, with the prefix held identical: the financial
suffix kills the riparian branch (mean $+0.001$; both
mixture sub-modes on the financial side; riverside bridges nearly
absent), and the riparian suffix kills the financial branch (mean
$-0.223$). This is \S\ref{ssec:polysemy}'s account of disambiguation,
measured: enriching the boundary does not erase the fork --- it poses
a \emph{different} horn whose foreclosed reading is no longer
generable. In TTS terms, the other branch's triangles cease to be
inhabited; in ML terms, right context prunes the bridge distribution
to one basin.

\paragraph{Two instruments, one horn.} Rerunning the same three horns
through native FIM ($n = 200$ each; logs in the supplement and
Appendix~\ref{app:logs}) gives a textbook instance of the paper's own
doctrine, and we set it out in full. First, what \textbf{replicates}:
the suffix effect, the theoretically central prediction, holds under
both instruments --- the condition means order identically, with the
neutral horn between the two disambiguated ones (native: NEU $-0.160$
between FIN $-0.054$ and RIV $-0.218$; prompted: NEU $-0.145$ between
FIN $+0.001$ and RIV $-0.223$). Whatever else the instruments
disagree about, both report that the right context prunes the bridge
distribution toward its own basin. Second, what \textbf{does not}:
the fork verdict at the neutral horn. Prompted FIM is decisively
bimodal ($\Delta\mathrm{BIC} = -17.8$); native FIM is unimodal-broad
($\Delta\mathrm{BIC} = +3.9$). The logs explain the disagreement.
The native bridges, produced by an objective trained for document
continuity, frequently bridge the \emph{scene} without re-engaging
the ambiguous word at all --- Tai Chi gatherings, a Mekong crossing,
a Ferris wheel against the sunset (Appendix~\ref{app:logs}) ---
deferral at scale: fluent connective tissue that postpones the
reading indefinitely. The prompted instrument's demand that the whole
``read as one coherent whole'' forces an actual \emph{reading}, and
readings commit. Third, what the disagreement is and is not. ``Instrument'' bundles
three roles the formalism keeps separate: the \emph{generator}
(which fillers become observable --- here, two different conditional
distributions), the \emph{readout} $\rho$ (how fillers are mapped
into a measured space), and the \emph{regime} (the
connection/separation verdicts in that space).
Theorem~\ref{thm:crossregime} formally governs the third role only
--- two regimes judging \emph{one} pair --- whereas native and
prompted FIM differ at the \emph{generator}: they sample different
bridge populations, so the theorem's configuration is not literally
instantiated. Nor is the native result a connection verdict: a broad
unimodal sample is failure-to-detect, and the calculus itself
insists failure-to-detect is not canonicity
(Remark~\ref{rem:canonunder}). What the comparison shows is the
theorem's \emph{lesson} operating one layer up --- verdicts are
relative to the whole instrument stack, and disagreement across
stacks is data about the stacks. The robust finding is what survives
both: the suffix effect. The instrument-validity question for the
full study is now sharply posed --- if the fork judgment's subject is
the space of \emph{readings}, the generator and readout of record
must be the ones that elicit and preserve readings rather than
scene-continuity, and that choice must be argued and pre-registered,
not slipped in.

\begin{figure}[t]
\centering
\includegraphics[width=0.9\linewidth]{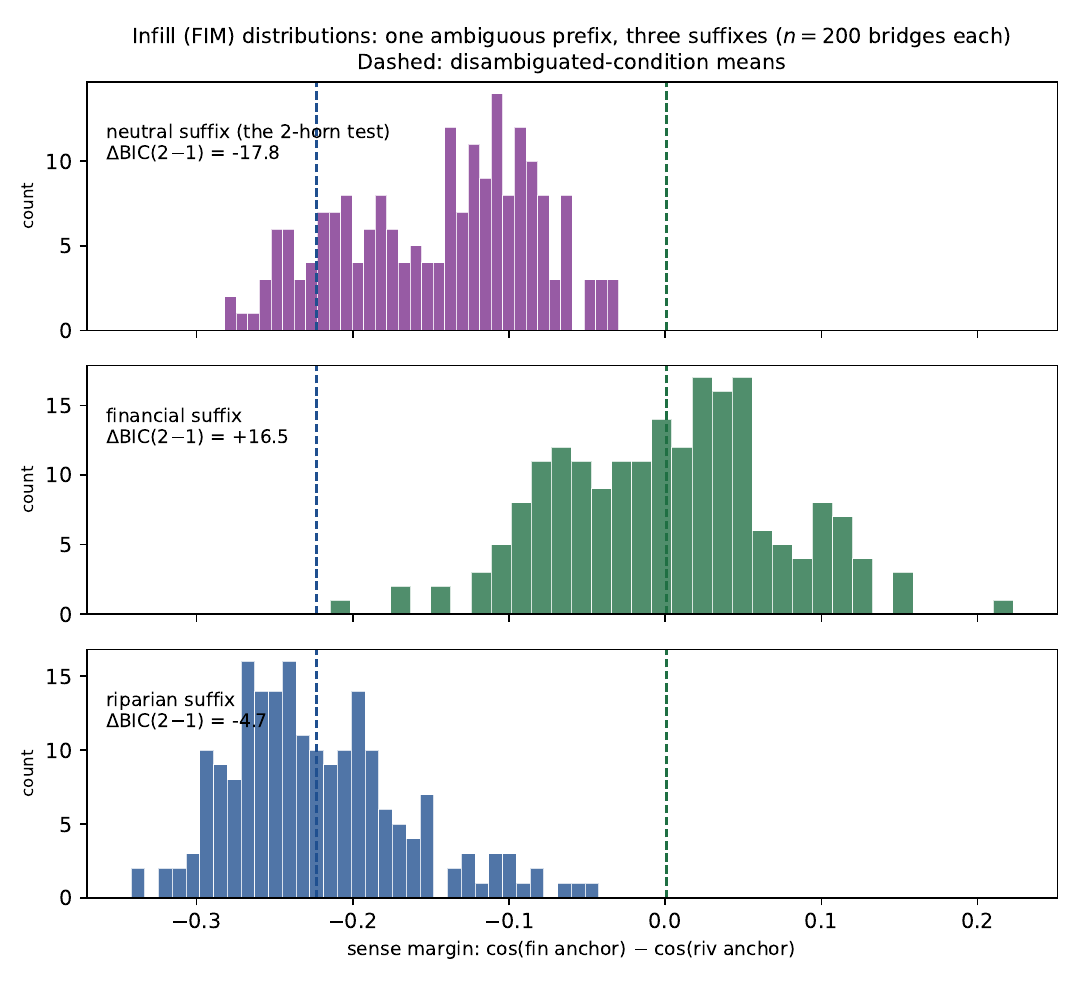}
\caption{The $2$-horn (FIM) experiment: one ambiguous prefix, three
suffixes, $n = 200$ sampled bridges per horn. Dashed lines: the
disambiguated conditions' means. The neutral-suffix horn (top) is
bimodal --- riverside-committed mode plus a deferral mass; each
disambiguating suffix (middle, bottom) collapses the bridge
distribution onto its own basin: the suffix effect.}
\label{fig:horn}
\end{figure}

\paragraph{For readers from machine learning.} The entire protocol
uses standard artifacts, and the theory is a logic over exactly those
artifacts. Resampling a prompt = sampling the filler space of a cone;
FIM with both flanks fixed = sampling an inner $2$-horn; the embedding
model with a threshold = a regime; the clustering output at swept
thresholds = the measurement context $\MC_P$; stable multimodality of
the completion or infill distribution = the fork; unimodality robust
to resolution = (estimated) canonicity; the persistence criterion =
the rejection region. What the calculus adds to this familiar
pipeline is a guarantee layer: verdicts compose across instruments
and resolutions by rules with a soundness theorem ---
separations persist under refinement, instruments may disagree only
in the configurations of Theorem~\ref{thm:crossregime}, and no
verdict appears that the recorded data do not force.

\paragraph{Honesty clause.} What is proven empirically today is
non-manifoldness and outcome-splitting; \emph{non-contractibility of
filler spaces} is strictly stronger than either, and no published test
computes local homotopy type. The fork test is the proposed bridge,
and it is falsifiable in both directions: ubiquitous single-component
results under refinement would disconfirm the fork reading of
forking tokens; stable multi-component results would instantiate the
theory's central judgment in measurement.

\section{Related work}\label{sec:related}

\paragraph{Synthetic $\infty$-category theory.} Simplicial type theory
\citep{riehlshulman2017}, its triangulated extension with directed
univalence \citep{gratzer2024}, machine-checked developments in Rzk
\citep{kudasov2024, lossin2026}, and the model-theoretic analysis of
\citep{rasekh2025} constitute a mature discipline of \emph{directed}
type theory. Throughout that lineage, composition-uniqueness
(Segal-ness) is a property to be \emph{imposed or established} for the
types of interest. TTS inverts the stance: Segal-ness is a local
judgment whose presence, failure, and boundary structure are the
objects of the logic. We use the lineage's vocabulary (and intend its
toolchain, \S\ref{sec:conclusion}) while occupying a complementary
question.

\paragraph{Two-level type theory.} 2LTT \citep{annenkov2023} pairs a
fibrant (HoTT) layer with a non-fibrant (exo-)layer; exo-types may
lack transport. The contrast with TTS is the contrast between
\emph{absence} and \emph{witness}: in 2LTT a non-fibrant type simply
fails to support the operations, and nothing internal certifies,
locates, or measures the failure; in TTS the failure of canonicity at
a specific horn at a specific regime is a positively inhabited type
($\Fork_V$) whose inhabitant is a term --- two named fillers and a
recorded separation --- with a provenance theorem
(Theorem~\ref{thm:provenance}) governing where such terms can come
from. The separation is term-level: fork witnesses persist under
refinement (Theorem~\ref{thm:persist}) and transport along readout
pullbacks (Proposition~\ref{prop:pullback}) because there is a term to
transport; 2LTT's non-fibrancy, being mere absence, supports no
analogous calculus of verdicts.

\paragraph{Categorical and dynamical accounts of learned
representations.} GAIA \citep{mahadevan2024} reads learning problems
as horn-extension problems over simplicial sets --- backpropagation
solves inner horns --- but is solution-seeking by design: its Kan
assumption exists so that extensions succeed, no failure is tracked,
and no theory of meaning is offered. Dynamical proposals model
generation as trajectory evolution on an assumed semantic manifold
\citep{zhang2025dmet}, precisely the assumption the measurements of
\S\ref{sec:empirics} defeat; none offers judgments, rules, or
metatheory. The formal-geometry result of \citep{jiang2024linear}
(softmax + implicit bias provably promote linear concept
representations within a latent-variable model) is the kind of
peer-reviewed anchor a semantics of embedding geometry can build on,
and is consistent with basin-interior linearity; it says nothing about
the singular loci that concern us. Whether LLM geometry is
\emph{semantic} at all is contested \citep{benderkoller2020,
lakemurphy2023}; TTS does not presuppose an answer --- its verdicts
concern the structure of the geometry under measurement, and the Frege
reading of \S\ref{sec:sense} is the philosophical wager that this
structure is where sense lives.

\paragraph{Theories of sense.} The procedural lineage --- Tich\'y/TIL
\citep{duzi2010}, Moschovakis \citep{moschovakis2006}, Muskens
\citep{muskens2005}, and in Martin-L\"of type theory Bentzen
\citep{bentzen2020} --- individuates sense by identity of procedures;
the proof-theoretic lineage \citep{tranchini2016, ayhan2021, ayhan2025}
by derivational structure, with \citep{ayhan2025}'s bilateralist
calculus of proofs \emph{and refutations} particularly close in
spirit to our positive/positive (canonicity/fork) architecture.
Truthmaker semantics \citep{fine2017} individuates hyperintensionally
by exact verifier sets, set-theoretically and without procedural or
geometric structure. The philosophical literature on identity in HoTT
\citep{ladymanpresnell2017, corfield2020} has the path-structure but
has not, to our knowledge, connected filler spaces to modes of
presentation. Informal anticipations that LLM latent geometry plays
the role of Fregean sense exist \citep{schuele2024}; TTS is, to our
knowledge, the first \emph{formal} system in which the identification
is carried by a calculus with a soundness theorem. Categorical
compositional distributional semantics \citep{coecke2010} is the other
formal bridge between compositional structure and vector geometry; the
stances are complementary --- DisCoCat composes meanings \emph{in} the
vector space, while TTS judges the failure modes of composition.
Within formal
semantics of natural language, the vector-space bridge of
\citep{quigley2025} (a homomorphism between Montagovian and
distributional semantics, in this journal) is methodologically the
closest published relative: it connects a formal semantics to
embedding geometry; TTS connects a \emph{proof theory} to that
geometry's failure modes.

\paragraph{Singular thought.} The mental-files tradition
\citep{recanati2012} is the closest philosophical model of
\S\ref{ssec:names}, and the alignment deserves precise statement ---
in our vocabulary, not theirs. What that tradition posits as primitive
mental particulars (files, dossiers) appears here as emergent
structure: regions of a presentation space. Its \emph{presumption of
co-reference} --- information within one file combines freely ---
corresponds to a monotonicity entailment: when the presentational
regime refines the objectual one, co-location in a presentational
component entails co-reference at the objectual regime
(\textsc{$\simeq$-trans} + \textsc{$\simeq$-mono}). Its
\emph{epistemically rewarding relations}, which individuate files by
the channel through which information arrives, correspond to regimes
--- an ER relation is an instrument. And its signature doctrine, that
identity discovery \emph{links} files without \emph{merging} them, is
Proposition~\ref{prop:namestransfer}: flow at the discovering regime,
provable segregation at the presentational one. We adopt the
alignment, not the vocabulary or the representationalist metaphysics:
nothing in \S\ref{ssec:names} is a stored particular; everything is
graded geometry under measurement.

\section{Conclusion and open problems}\label{sec:conclusion}

We have presented a type theory in which the semantically decisive
events of a stratified meaning-space --- forced composition, genuine
choice, locked fluency --- are graded verdicts with positive,
term-level witnesses; in which the verdicts are consequences of
measurement by construction and by theorem (conservativity,
no-fork-ex-nihilo, provenance, record consistency); and in which
Frege's distinction is theorem-grade structure: sense is choice of
filler, hyperintensional distinctness is apartness at a regime, and
the consistency of discovered identity with standing difference is
characterized exactly by the order between the instruments involved.
Exclusion is derived, not legislated; the witnesses are measurements,
not stipulations; and what the calculus cannot do --- certify
canonicity from a finite record, manufacture a fork, retract a verdict
--- is as much a part of its content as what it can.

Open problems, in ascending order of ambition:
\begin{enumerate}[label=(\arabic*)]
\item \textbf{Names as expressions.} \S\ref{ssec:names} delivers the
  name-\emph{state} configuration; what remains is the public name as
  an expression --- the word-to-constant link, predication dynamics
  over presentation-points, and rigidity in full
  (Remark~\ref{rem:rigidity}). A first name-pilot, logged in the
  supplement, taught us more than its design intended. The
  hypothesized dissociation --- names apart presentationally,
  connected under predication readouts --- did not appear for
  \texttt{gpt-3.5-turbo-instruct}. The logs show why. Asked for
  facts, the model answers about Venus perfectly: ``108 million
  km,'' every draw. It never routes the classical names through that
  knowledge: Phosphorus's facts are confabulated. Yet asked
  directly, it affirms the identity ($0.90$ yes-same-object). So a
  one-step Leibniz syllogism fails in a fully \emph{transparent}
  context: both premises assented to, the substitution never
  executed. That is stronger than Frege's puzzle, which concerns
  opaque contexts. The calculus names the situation precisely.
  Substitution everywhere is transport along ambient identity.
  Ambient identity is refuted by the standing presentational
  separation (id-clash, Corollary~\ref{cor:idclash}). What the model
  holds is the graded connection $\simV{W}$, and that licenses
  substitution at $W$ alone
  (Proposition~\ref{prop:namestransfer}). Predication and
  presentation are, for these signifiers, distinct regimes that do
  not marry. Meaning here does not operate syllogistically: the
  identities it trades in are never $\mathsf{Id}$-grade, only some
  instrument's verdict. (Compare the reversal curse
  \citep{berglund2023}: factual knowledge not closed under one-step
  logical transforms.) The question for the full study: which
  inferential closures does a model's geometry implement, and at
  which regimes?
\item \textbf{General horns.} Extend \S\ref{sec:grades} to
  $\Lambda^n_i$ (Remark~\ref{rem:higher}); identify which semantic
  phenomena require $n \ge 3$ (anaphora across discourse units is a
  candidate).
\item \textbf{The alignment--diversity question.} The $2$-horn
  experiment's instrument finding --- an instruction-tuned model's
  bridge distribution is mode-collapsed where the completion model's
  forks --- deserves systematic study: does alignment tuning collapse
  the filler distributions (ever-stronger failure-to-refute
  canonicity --- no record can do more, Remark~\ref{rem:canonunder}),
  and can the fork census quantify what it removes?
\item \textbf{Normalization and decidability.} The extension is
  axiomatic; canonicity and normalization for it are not claimed.
  The natural hardening is fragment-wise: the relational fragment's
  proof terms carry near-groupoid structure ($\simV{V}$-terms under
  sym/trans, $\apV{V}$-terms acted on by \textsc{ext}), suggesting a
  confluent rewriting independent of the (axiomatically inert)
  directed fragment --- which would yield \emph{syntactic} provenance
  exactly where it matters, for the fork witnesses, without solving
  normalization for the whole theory.
\item \textbf{Completeness.} Soundness is Theorem~\ref{thm:soundness};
  completeness for regime-indexed simplicial apartness spaces is open.
  We conjecture it fails for the full theory (the directed fragment is
  axiomatically inert) while holding for the relational fragment over
  resolution models on finite carriers.
\item \textbf{Mechanization.} The directed fragment is within reach of
  the Rzk ecosystem \citep{kudasov2024}; the relational fragment is
  not, and an Agda development with postulated regimes and apartness
  is the pragmatic first step.
\item \textbf{The empirical program.} The fork test at scale: a census
  of the fork locus of a production model, the singular-token
  cross-check against \citet{robinson2025manifold}, and the diachronic
  question --- how the locus moves under fine-tuning --- which is
  where this logic touches alignment.
\end{enumerate}

\appendix
\section{Proofs}\label{app:soundness}

\subsection{Base consistency}

\begin{lemma}[Base consistency]\label{lem:baseconsistency}
MLTT${}+{}$hom/$\Tri$ over any semantic signature $\Sigma$ is consistent.
\end{lemma}
\begin{proof}
Extend any model of MLTT (e.g.\ in $\mathsf{Set}$) by interpreting each
base semantic type as $\mathbf{1}$, each point/hom/$\Tri$ constant of
$\Sigma$ as $\star$, and each function constant (readouts,
Definition~\ref{def:readout}) as the unique map into $\mathbf{1}$; set
$\hom{}{a}{b} := \mathbf{1}$ and $\Tri(f,g;h) := \mathbf{1}$,
with $\mathsf{id}_a, \lambda^{\mathsf{l}}_f, \lambda^{\mathsf{r}}_f
\mapsto \star$. The interpretation is compositional; the directed fragment
has no eliminators and no computation rules, so all five of its rules are
validated, and $\emp$ remains uninhabited. (The path-action $e_* f$ is
derived by transport and needs no separate clause.)
\end{proof}

\subsection{Erasure and conservativity}

\begin{lemma}[Erasure soundness]\label{lem:erasure}
Define $(-)^{\circ}$ on the language of $\mathrm{TTS}(\Sigma, \emptyset)$
by $(u \apV{V} v)^{\circ} := \emp$, $(u \simV{V} v)^{\circ} :=
\mathbf{1}$, homomorphically elsewhere (contexts, types, terms,
derivations). Then $\Gamma \vdash_{\emptyset} t : A$ implies
$\Gamma^{\circ} \vdash t^{\circ} : A^{\circ}$ in MLTT${}+{}$hom/$\Tri$.
\end{lemma}
\begin{proof}
By induction on derivations; the MLTT and directed rules are untouched by
the translation. The relational rules are admissible in the image:
formation rules become ``$\mathbf{1}\ \mathsf{type}$'' /
``$\emp\ \mathsf{type}$''; the formers' built-in proposition-hood becomes
the standard proofs that $\mathbf{1}$ and $\emp$ are h-propositions;
\textsc{$\simeq$-refl} $\mapsto \star$;
\textsc{$\simeq$-sym/trans/mono} $\mapsto$ constants valued $\star$;
\textsc{$\#$-irrefl/sym/mono} $\mapsto \mathsf{id}_{\emp}$;
\textsc{$\#$-cotrans} $\mapsto$ $\emp$-elimination into
$\|\emp + \emp\|$; \textsc{$\#$-ext} $\mapsto
\lambda (x{:}\mathbf{1})(y{:}\emp).\, y$; id-to-regime $\mapsto$ constant
$\star$. The translation preserves judgmental equality since the erased
formers carry no computation rules. Note ($\MC$-ax) is precisely the rule
erasure cannot validate --- a $\#$-entry would erase to a closed term of
$\emp$ --- which is why the lemma is stated for the empty record.
\end{proof}

\begin{proof}[Proof of Theorem~\ref{thm:conserv}]
For $\Gamma, A$ in the $\simeq/\#$-free fragment, $\Gamma^{\circ} =
\Gamma$ and $A^{\circ} = A$; apply Lemma~\ref{lem:erasure}.
\end{proof}

\begin{proof}[Proof of Corollary~\ref{cor:noforkexnihilo}]
Suppose $\vdash_{\emptyset} t : \Fork_V(\sigma)$. Erasure gives
$\vdash t^{\circ} : \sum_{\varphi, \psi : (\Fill\sigma)^{\circ}} \emp$;
projecting twice yields a closed term of $\emp$, contradicting
Lemma~\ref{lem:baseconsistency}. For the judgment of record: erasure sends
$\|\Fork_V(\sigma)\|$ to $\|\Sigma_{\varphi\psi}\,\emp\|$, and the
truncation eliminates into the h-proposition $\emp$ via the projections,
giving the same contradiction. The companion fact --- no closed
$\vdash_{\emptyset} t : (u \simV{V} v) \to \emp$ --- follows by applying
$t^{\circ} : \mathbf{1} \to \emp$ to $\star$: the empty record refutes no
indiscernibility either.
\end{proof}

\begin{proof}[Proof of Theorem~\ref{thm:provenance}]
Immediate from Theorem~\ref{thm:soundness} applied to the closed judgment:
every model realizing $\MC$ inhabits $\Sem{\Fork_V(\sigma)}$; the
contrapositive gives underivability whenever one realizing model fails to
fork $\sigma$ at $V$.
\end{proof}

\subsection{Soundness: the case analysis}

\begin{proof}[Proof of Theorem~\ref{thm:soundness}]
Induction on derivations, over a model per Definition~\ref{def:sas}
realizing $\MC$.
\emph{MLTT fragment.} Standard \citep{hottbook}, with one global
simplification: by (M3) the identity type is interpreted discretely, so
every instance of path induction is a reflexivity-transport ---
$\Sem{J(c; \refl)} = \Sem{c}$ --- and the case closes for \emph{arbitrary}
motives, including families built from $\mathsf{hom}$, $\Tri$,
$\simV{V}$ and $\apV{V}$. (This is where coarse path-object
interpretations would fail: a non-trivial identity between $a, b$ with
$\hom{}{a}{c}$ inhabited and $\hom{}{b}{c}$ empty refutes the transport
case; such models are excluded by (M3).)
\emph{Directed fragment.} $\Sem{\hom{A}{a}{b}}$ is the fibre of
$(d_1, d_0) : X_1 \to X_0 \times X_0$; $\Sem{\Tri}$ the fibre of the
boundary map $X_2 \to X_1^3$; degeneracies interpret $\mathsf{id}$,
$\lambda^{\mathsf{l}}$, $\lambda^{\mathsf{r}}$ via the simplicial
identities ($d_2 s_0 f = \mathsf{id}$, $d_0 s_0 f = f = d_1 s_0 f$, and
dually for $s_1$). No rule demands inhabitation
of any fibre: the fragment is sound in any simplicial object, Kan or not,
Segal or not.
\emph{Relational fragment.} Per (M1), each denotation carries
$(\simV{V}, \apV{V})$ satisfying the rules of \S\ref{ssec:apart}; each
rule is validated pointwise. \textsc{$\simeq$/$\#$-mono} hold by the
regime-indexed structure; \textsc{$\#$-ext} and the component laws were
verified for the resolution model in Example~\ref{ex:resolution}.
\emph{Record.} ($\MC$-ax) holds by realization.
\emph{Substitution.} The new formers commute with substitution by
convention; coherence of the interpretation follows the standard
comprehension-category treatment of the MLTT fragment, extended to the
new formers componentwise; under (M3) no fibration structure beyond
reflexivity-transport is required.
\end{proof}

\subsection{Cross-regime witnesses}

\begin{proof}[Witnesses for Theorem~\ref{thm:crossregime}(b), in full]
\emph{Case $V < W$ (strictly finer separator).} Resolution model on a
two-point carrier $\{u, v\}$ at distance $r$ with
$\varepsilon_V < r \le \varepsilon_W$ (such $r$ exists by strictness).
Then $u \simV{\varepsilon_W} v$ by the single $\le \varepsilon_W$ step;
$u \apV{\varepsilon_V} v$ since $r > \varepsilon_V$; and the cross-regime
monotonicity laws hold: $\simV{\varepsilon_V}$ is the diagonal, contained
in the total $\simV{\varepsilon_W}$, and $\apV{\varepsilon_W}$ is empty,
contained in the pair $\apV{\varepsilon_V}$.
\emph{Case $V, W$ incomparable.} The two-point, two-regime record model:
carrier $\{u, v\}$; $\apV{V}$ the symmetric pair, $\simV{V}$ the diagonal;
$\simV{W}$ total, $\apV{W}$ empty. Every rule of \S\ref{ssec:apart}
holds per regime: at $V$, cotransitivity splits within $\{u, v\}$ (one
disjunct is always the recorded pair) and extensionality is trivial on the
diagonal; at $W$ the apartness rules are vacuous.
\emph{Extension to the full poset.} For any regime $U$ set
$\simV{U} := $ total if $W \le U$, else diagonal; $\apV{U} := $ the pair
if $U \le V$, else empty. Monotonicity is transitivity of $\le$; a clash
at some $U$ would need $\simV{U}$ total and $\apV{U}$ nonempty, i.e.\
$W \le U \le V$, contradicting $W \not\le V$.
\emph{Completion to a full model.} Both witnesses extend to
Definition~\ref{def:sas} as the constant simplicial object on the
carrier (faces and degeneracies identities, identity discrete), with
hom- and $\Tri$-fibres taken as the genuine fibres of the boundary
maps --- the diagonal fibres are singletons, the off-diagonal fibres
empty, which is harmless: no rule demands inhabitation of any fibre.
(M2) and (M3) hold trivially. The rest of the signature is interpreted
as in Lemma~\ref{lem:baseconsistency} (other base types as
$\mathbf{1}$, constants as $\star$, function constants as the unique
map); further constants of the carrier type may be sent to either
point, since the record mentions only the named pair.
\end{proof}

\subsection{The two-point diagonal model}

\begin{proof}[Model for Proposition~\ref{prop:nameinfo}(ii)]
Interpret $P$ as the two-point carrier $\{x, y\}$ with $\Sem{h} = x$,
$\Sem{p} = y$; at \emph{every} regime $U$ set $\simV{U} :=$ the
diagonal and $\apV{U} := \emptyset$; complete to
Definition~\ref{def:sas} as the constant simplicial object with
discrete identity, interpreting the rest of the signature as in
Lemma~\ref{lem:baseconsistency} (freshness of $h, p$ --- no other
$\Sigma$-declaration or record entry mentions them --- is what makes
this free). Every rule of \S\ref{ssec:apart} holds:
refl/sym/trans on the diagonal; \textsc{$\simeq$-mono} across any
poset pair is diagonal $\subseteq$ diagonal; the five apartness rules
are vacuous on the empty relation; id-to-regime lands on the diagonal
under discrete identity. The model realizes the empty record
vacuously, and both $\Sem{h \simV{W} p}$ and $\Sem{h =_P p}$ are
empty at $(x, y)$; Theorem~\ref{thm:soundness} gives underivability
of both.
\end{proof}
\section*{Acknowledgments}
The formal system, its verification, the experimental design and the
prose of this paper were developed in sustained collaboration with
Nahla, an AI research collaborator (Claude, Anthropic), whose
contributions span the apartness-grounding of the fork judgment, the
measurement-context architecture, the metatheorems, and the pilot
study. Per publisher policy on AI authorship, responsibility for the
content rests with the human author.

\input{appendix-logs}

\end{document}

%% file: appendix-logs.tex
\section{Experiment log excerpts}\label{app:logs}

Every text sampled in this paper is logged verbatim, with its full prompt context, in the supplementary data package (\S\ref{sec:empirics}, data availability). This appendix reproduces, for each key condition, the exact prompt and a seeded random sample of the logged outputs ($\mathrm{seed}=11235$, indices shown) --- unabridged and uncorrected, exactly as generated. The reader can check every quoted text against the package by its index.

\subsection*{Free-completion pilot (\texttt{gpt-3.5-turbo-instruct}, $T=1.0$, $n=300$/condition)}
Each record: the prompt below, completed freely; first sentences of the completions were embedded. The four completions of Example~\ref{ex:miniature} are records \#114, \#64, \#89 and \#199 of the AMB log.

\subsection*{Cone / AMB (calibrated ambiguous boundary)}

\noindent Prompt (verbatim): \textit{``We met them near the bank, just before sunset.''}
\begin{itemize}[itemsep=2pt]
\item[\#20:] \textit{The clouds were beginning to turn a deep purple and the air was thick with the smell of rain.\;\textsf{\small[newline]}\;\;\textsf{\small[newline]}\;"Welcome to our village," the man said, extending his hand in greeting. He introduced himself as}
\item[\#153:] \textit{Raymond was the father of the only person we knew, which made this strange meeting even more bizarre. Aaron seemed to vanish with the sunlight, but Raymond stood in the shadows, waiting for us.}
\item[\#155:] \textit{They had gotten out of a sleek black car, their faces hidden behind sunglasses. Five of them, all dressed in sharp suits, looking more like secret agents than the vigilantes they purported to be.}
\item[\#170:] \textit{It was a group of four, two men and two women, all wearing dark clothing and masks. They were armed with guns and led us inside the bank, where they quickly took control of the situation}
\item[\#185:] \textit{Like the early afternoon when these kinds of joints are crowded, the evening before the banks close proved to be favorable for moving cryptocurrency.\;\textsf{\small[newline]}\;\;\textsf{\small[newline]}\;These kinds of operations are risky and complicated. Often done with unprofession}
\item[\#192:] \textit{I can’t tell you whether it was gonna be a clout session or a crushing session or something new, but they were gonna do something, and it was gonna be soon.\;\textsf{\small[newline]}\;\;\textsf{\small[newline]}\;“Sister}
\item[\#199:] \textit{As we paddled our kayaks up to the bank, we see that they were waiting for us. There were three of them, all dressed in simple grey robes, with their heads bald and their}
\item[\#210:] \textit{Sarah was carrying several books in her arms and Matt had a sketch pad tucked under his arm. They were both talking animatedly about a book they had just finished reading and were eager to share their thoughts}
\end{itemize}

\subsection*{Cone / FIN control}

\noindent Prompt (verbatim): \textit{``He needed to deposit the check before noon, so he walked down to the bank.''}
\begin{itemize}[itemsep=2pt]
\item[\#20:] \textit{As he entered the bank, he was greeted by a friendly teller who directed him to the deposit slips. He quickly filled one out with his information and the amount of the check.\;\textsf{\small[newline]}\;\;\textsf{\small[newline]}\;He handed}
\item[\#153:] \textit{As he made his way through the crowded streets, he couldn't help but feel a sense of nostalgia wash over him. He had grown up in this small town and had spent countless afternoons}
\item[\#155:] \textit{As he approached the bank, he noticed a line forming outside the door. It seemed that he wasn't the only one with important errands to run before noon. He joined the line, impatient}
\item[\#170:] \textit{As he entered the bank, he noticed a long line of customers waiting to be served. He groaned inwardly, knowing that he would have to wait for a while.\;\textsf{\small[newline]}\;\;\textsf{\small[newline]}\;He took a number}
\item[\#185:] \textit{As he entered the bank, he noticed that the line was long and he didn't have much time left. He approached the teller and handed her the check, explaining that he needed to deposit}
\item[\#192:] \textit{As he entered the bank, he was greeted by the teller who asked how she could assist him. He explained that he needed to deposit a check before noon.\;\textsf{\small[newline]}\;\;\textsf{\small[newline]}\;The teller directed him to}
\item[\#199:] \textit{As he stood in line, he checked his watch anxiously. It was already eleven thirty and he could feel the time slipping away. He hoped the line would move quickly.\;\textsf{\small[newline]}\;\;\textsf{\small[newline]}\;Finally, it was}
\item[\#210:] \textit{As he approached the bank, he noticed a long line of people outside the entrance. He groaned, knowing he would have to wait in line for a while before he could deposit his check.}
\end{itemize}

\subsection*{Cone / RIV control}

\noindent Prompt (verbatim): \textit{``He carried his fishing rod through the reeds and walked down to the bank.''}
\begin{itemize}[itemsep=2pt]
\item[\#20:] \textit{The sound of water gently lapping against the shore soothed his soul. He took a deep breath and cast his line. \;\textsf{\small[newline]}\;\;\textsf{\small[newline]}\;The sun was just beginning to rise, casting a warm glow}
\item[\#153:] \textit{Even as he stepped onto the muddy shore, he felt the tension of everyday life ease from his shoulders. This was his sanctuary, his escape from reality. He breathed in the fresh air and let the}
\item[\#155:] \textit{At this time of year, his favorite spot on the Mize River was shaded by a pair of ancient willow trees that dipped their twisted branches into the water. The air was thick and humid and}
\item[\#170:] \textit{The sun was just starting to dip below the horizon, casting a warm orange glow over the water. He took a deep breath, taking in the fresh, crisp air and the peacefulness of the}
\item[\#185:] \textit{I stopped just short of where he was standing. If he didn't want my help that was fine. I silently took note of how he was rigged: on the end of his line he had a}
\item[\#192:] \textit{He cast his line out into the water and settled down to wait. As he sat there, he took in the peaceful surroundings – the sound of the birds chirping, the gentle rustling of the}
\item[\#199:] \textit{He smiled as he felt a sense of calm wash over him. The sun was shining and there was a light breeze that rustled the leaves on the trees above. He could hear the sound of birds}
\item[\#210:] \textit{Ducks were lazing on the pond’s surface and would scatter as he edged closer. It was a roasting, sunny day and it was perfect. Cormac let out a pleased breath and calmly cast}
\end{itemize}

\subsection*{The $2$-horn (FIM) experiment (\texttt{gpt-3.5-turbo-instruct}, prompted infilling, $T=1.0$, $n=200$/horn)}
Each record: prefix and suffix below, both held fixed; the model samples the bridge. The FIM template is reproduced in the supplementary package and in \texttt{run\_horn.py}.

\subsection*{Horn / neutral suffix --- the $2$-horn fork test}

\noindent Prefix: \textit{``We met them near the bank, just before sunset.''} \quad Suffix: \textit{``By the time we said our goodbyes, the stars were already out.''}
\begin{itemize}[itemsep=2pt]
\item[\#77:] \textit{As we walked along the riverbank, we shared stories and laughter with our new friends, watching the sky turn shades of pink and orange as the sun began to set behind us. We exchanged contact information and promised to keep in touch before parting ways, already reminiscing on the wonderful evening spent together.}
\item[\#85:] \textit{As we chatted with the group of strangers, we learned that they were travelers from a different country, exploring the same city as us. They shared stories of their adventures and we shared tips for the best local spots to visit.}
\item[\#92:] \textit{They were a merry group, with smiles on their faces and laughter in their voices. We shared stories and drinks by the riverbank, enjoying the warm summer air.}
\item[\#96:] \textit{They were a group of travelers from a far-off land, telling tales of their adventures and sharing exotic spices for our food. As the night went on, we laughed and learned from each other, creating bonds that would last long after we parted ways.}
\item[\#99:] \textit{We chatted and laughed with them, swapping stories and sharing snacks. As the sun began to dip below the horizon, we made plans to meet up again tomorrow for a hike in the mountains.}
\item[\#162:] \textit{We had a picnic by the river, laughing and sharing stories as the sun set over the water. They told us about their travels and we shared our plans for the future. It felt like we had known each other for years, and we promised to stay in touch.}
\item[\#189:] \textit{They were a group of travelers passing through our small town, and we quickly struck up a conversation. Their stories of faraway lands and thrilling adventures captivated us, and we ended up sharing a meal together at a nearby tavern. Laughter and camaraderie filled the air as we swapped tales late into the night.}
\item[\#196:] \textit{They were a group of local fishermen, and they invited us to join them for a drink on their boat. As we sat there, sipping on cold beer and listening to their stories of life at sea, the sun slowly dipped below the horizon and painted the sky in stunning shades of orange and purple. It was a peaceful and memorable evening, and}
\end{itemize}

\subsection*{Horn / financial suffix}

\noindent Prefix: \textit{``We met them near the bank, just before sunset.''} \quad Suffix: \textit{``When the manager finally unlocked the doors again, we were the last customers of the day.''}
\begin{itemize}[itemsep=2pt]
\item[\#77:] \textit{The bank was bustling as people hurriedly finished their transactions before the doors closed. We noticed a couple standing nearby, looking confused and anxious. Without thinking, we went over to them and asked if they needed any help. They explained that they had just moved to the city and were trying to open a new account, but had missed the bank's}
\item[\#85:] \textit{They were a group of young kids, no more than ten or eleven years old, sitting on the sidewalk with a cardboard sign that read "Need help, anything will do." My heart broke seeing them there, so I asked them if they were hungry. They eagerly nodded and I promised to buy them a meal. We all went inside the bank together}
\item[\#92:] \textit{As we walked out, we realized we had forgotten to withdraw cash, so we went back in. The bank was now nearly empty except for a few tellers and the security guard who greeted us with a friendly smile.}
\item[\#96:] \textit{As we stepped outside, we noticed a group of street performers had set up nearby. They were juggling, dancing and playing music, drawing a small crowd with their lively energy. We paused to watch for a moment before heading to our cars, feeling grateful for the unexpected entertainment after a long day at work.}
\item[\#99:] \textit{We chatted with them briefly, finding out they were new to town and looking to set up a small business. They seemed like a friendly couple, and we promised to stop by their shop once they were up and running. As we walked away, we couldn't help but feel excited for this new addition to our community.}
\item[\#162:] \textit{As we walked out with our newly acquired loan, I couldn't help but reflect on how lucky we were to have stumbled upon the friendly couple who offered us a ride to the bank after our car broke down.}
\item[\#189:] \textit{We had been standing in line for what felt like hours, our stomachs grumbling and our feet sore. As soon as we saw the "Open" sign flicker on, we rushed inside, eager to complete our errands before the bank closed for the day. The teller greeted us with a smile and processed our transactions efficiently, but as}
\item[\#196:] \textit{The bank's interior was cool and dimly lit, the air thick with the smell of money. We quickly filled out the necessary paperwork and handed over the check, relieved to have our transaction completed. As we left, we noticed the sun sinking lower in the sky and decided to stay and watch the sunset from the nearby park.}
\end{itemize}

\subsection*{Horn / riparian suffix}

\noindent Prefix: \textit{``We met them near the bank, just before sunset.''} \quad Suffix: \textit{``We watched the current carry our little boat downstream until it disappeared.''}
\begin{itemize}[itemsep=2pt]
\item[\#77:] \textit{We had been planning this trip for weeks, and the river had seemed like the perfect spot to unwind and reconnect with nature. But as the sky darkened and the temperature dropped, we realized we were ill-prepared for the journey. The friendly couple we had met offered to tow us back to shore, and we gratefully accepted, feeling humbled}
\item[\#85:] \textit{The group of locals had warned us about the treacherous waters, but we didn't listen. As we struggled to navigate the rough waves and sharp rocks, we realized our mistake. But it was already too late. Our boat capsized and we were left clinging onto debris for dear life. Thankfully, the locals came to our rescue and brought us}
\item[\#92:] \textit{The strangers offered us a ride on their bigger boat and we accepted gratefully. As we traveled down the river, we shared stories and laughs, and before we knew it, the sun had set and the stars had come out, filling the sky with their twinkling lights.}
\item[\#96:] \textit{As the sun began to dip below the horizon, we decided to paddle back to shore. Suddenly, the wind picked up and the water became choppy, making it difficult for us to row against the current. We were worried we wouldn't make it back in time, but thankfully, we eventually reached the shore and pulled our boat onto the bank just}
\item[\#99:] \textit{We waved goodbye to our new friends and turned to head back to our campsite. As the sun sank below the horizon, we sat by the fire and recounted the day's adventures, grateful for the unexpected encounter that had made our trip even more memorable.}
\item[\#162:] \textit{As we waved goodbye, I couldn't help but feel a twinge of sadness. These strangers we had met just one day ago had shown us such kindness and hospitality on our journey down the river. I knew we would never forget the experience and the unexpected friends we had made.}
\item[\#189:] \textit{We saw them on the other side of the river, waving and laughing at us. We couldn't help but join in and wave back, feeling a sense of joy and connection even though we were separated by the rushing waters.}
\item[\#196:] \textit{As we made our way towards the shore, they offered to give us a ride on their larger boat. We eagerly accepted, grateful for the chance to rest our tired arms.}
\end{itemize}

\subsection*{Native FIM (API suffix conditioning, \texttt{gpt-3.5-turbo-instruct}, $T=1.0$, $n=200$/horn)}
Same horns as above, no instruction template: the prefix is the prompt, the suffix is passed natively. Note the scene-bridging register --- the basis of the two-instruments finding.

\subsection*{Native FIM / NEU}

\noindent Suffix: \textit{``By the time we said our goodbyes, the stars were already out.''}
\begin{itemize}[itemsep=2pt]
\item[\#77:] \textit{It was a sprawling group of all ages, practicing Tai Chi. Curiosity leading the way, we walked among them, listening to the sounds of the summer evening at the Chagrin River.}
\item[\#85:] \textit{It was cool and slowly cooling even more, but no one seemed to mind as they tightly wound terrible tales about corpses and skeletons and rivers that ran with untold horrors.\;\textsf{\small[newline]}\;\;\textsf{\small[newline]}\;Alex told me that he was a writer, that he was well known in The Realm. But I recognized him at once as the reality broker who’d started our little scavenger hunt}
\item[\#92:] \textit{This was the second day of our backup crossing of Mekong. We had started near Mae Salong, drove till Chiang Khong, crossed to HueXieng in Laos on the first day. And then drove south - to a village on the outskirts of Luang Prabang.\;\textsf{\small[newline]}\;\;\textsf{\small[newline]}\;We, a group of four, had 'saved'}
\item[\#96:] \textit{The band stood against the railing with the sun behind them, the buildings and—as fate would have it—the Ferris wheel (tourists) behind us. We got to talk music, inspiration and even sentiments of the gambler.}
\item[\#99:] \textit{In a quiet, sheltered inlet, Dave and Luticia were sitting in their small boat, patiently pulling in a fishing line. They must have spotted us paddling our kayak their direction. I’m glad they waved us down.\;\textsf{\small[newline]}\;\;\textsf{\small[newline]}\;They had a beautiful boat – probably one of the most graceful-looking I’ve ever seen. It was the kayak variety}
\item[\#162:] \textit{A bit shy at first, Grace quickly opened up, showing us around her community and introducing us to her daughter.}
\item[\#189:] \textit{The lights were dipping below the mountains, staining the sky violet and pink. They had wandered here from up north, and had brought their fiddles, their drums and their songs with them. Their voices were like water, spells woven together to pull light out of the dying day.}
\item[\#196:] \textit{By the time we had reached the palms, their boat had transformed from a onetime shipping vessel to a dinner party.\;\textsf{\small[newline]}\;\;\textsf{\small[newline]}\;Threadbare rugs were strewn over what seemed like it had previously been just plain old timber, a worn guitar had made its appearance, and plastic chairs had been set around a worn table. My aunt had already seated herself and booked}
\end{itemize}

\subsection*{Native FIM / FIN}

\noindent Suffix: \textit{``When the manager finally unlocked the doors again, we were the last customers of the day.''}
\begin{itemize}[itemsep=2pt]
\item[\#85:] \textit{They were on the water because the bank had turned them away. It was five minutes to closing time.\;\textsf{\small[newline]}\;\;\textsf{\small[newline]}\;Sitting in the back of Mark’s truck, watching the Swans in the warm evening breeze, was a pleasant way to wait a few minutes. And to them, that’s how we were. Just waiting a few minutes.\;\textsf{\small[newline]}\;\;\textsf{\small[newline]}\;I thought they}
\item[\#96:] \textit{Each was a dignitary in his people’s community, bearing gifts appropriate to a king’s reception. The meeting was as a very long week in the lives of a wife and husband away from home.}
\item[\#99:] \textit{They stood out awkwardly on the sidewalk, square, ugly, and utterly out of place. We had seen them before, and knew what they were capable of.\;\textsf{\small[newline]}\;\;\textsf{\small[newline]}\;“Didn’t know they’d spread so far.” I muttered under my breath. Rae gave me a stern look. We both knew our job, and complaints wouldn’t bring change.\;\textsf{\small[newline]}\;\;\textsf{\small[newline]}\;“They}
\item[\#196:] \textit{We smiled at them and motioned to go ahead, feeling good that we showed them courtesy. Later in line, they made the “gun” finger-in-the-pocket signal. No big deal, we thought, just kids horse-playing. We swiftly did our banking and went to our car. But we didn’t drive off. They passed by menacingly}
\end{itemize}

\subsection*{Native FIM / RIV}

\noindent Suffix: \textit{``We watched the current carry our little boat downstream until it disappeared.''}
\begin{itemize}[itemsep=2pt]
\item[\#85:] \textit{The little village we came across was quite unexpected, but at least we could spend the night in the tavern and make our way to the capital with daylight on our side.\;\textsf{\small[newline]}\;\;\textsf{\small[newline]}\;But despite Collin’s objections, we got caught in the moment and decided to take a boat ride in the river. It was the old boatman’s last day of work and}
\item[\#96:] \textit{They wanted to go home before the darkness descended, and there was an island that would be pitch black by seven pm. That meant crossing several streams, a village, and of course; the last of them - the most terrifying - was the river.\;\textsf{\small[newline]}\;\;\textsf{\small[newline]}\;"We can't let you cross this river, not in this darkness," said a helpful woman from the}
\item[\#99:] \textit{They were young, just a bit older than us, and they were working hard, clearly playing a game. They kept jumping into the river, swimming back to the shore, jumping in, over and over. They invited us to join, and said, "C'mon, if we're going to build this here raft we'll need some hard}
\item[\#196:] \textit{At that time of year, the river is little more than a collection of puddles. And with yam prices skyrocketing, these girls were taking no chances.}
\end{itemize}

\subsection*{The failed instrument: \texttt{gpt-4o-mini} calibration ($T=1.0$, $n=10$/horn)}
The basis of the mode-collapse finding (\S\ref{sec:empirics}): the instruction-tuned model's bridges, same horns as above. Note the near-identical sociable register and the weak uptake of the suffix constraint.

\subsection*{\texttt{gpt-4o-mini} / NEU}

\noindent Suffix: \textit{``By the time we said our goodbyes, the stars were already out.''}
\begin{itemize}[itemsep=2pt]
\item[\#5:] \textit{They were a group of friends, laughing and sharing stories as the day turned to dusk, their silhouettes casting long shadows on the water's surface.}
\item[\#6:] \textit{They were setting up a small campfire, laughter ringing in the air as they shared stories from their recent adventures.}
\item[\#8:] \textit{Their laughter echoed off the water as we shared stories, the sky slowly transitioning from hues of orange to deep navy.}
\item[\#9:] \textit{Their laughter echoed across the water as we shared stories and watched the sun dip below the horizon, painting the sky in shades of orange and purple.}
\end{itemize}

\subsection*{\texttt{gpt-4o-mini} / FIN}

\noindent Suffix: \textit{``When the manager finally unlocked the doors again, we were the last customers of the day.''}
\begin{itemize}[itemsep=2pt]
\item[\#5:] \textit{Their laughter echoed as they tossed pebbles into the water, the golden light dancing across the surface, and we were drawn to their carefree spirit.}
\item[\#6:] \textit{They were a small family, huddled close together, sharing stories of their day while casting long shadows on the grass. As we approached, their laughter filled the air, creating an inviting warmth that made us hesitate to leave.}
\item[\#8:] \textit{They were a friendly couple, eagerly sharing tales of their recent travels while we waited for the bank to reopen. As the sun dipped below the horizon, their laughter echoed in the cool evening air, making the short wait feel much longer.}
\item[\#9:] \textit{They had set up a small pop-up stall selling freshly made pastries, their laughter mixing with the gentle sounds of the river nearby. Intrigued, we couldn’t resist trying a few treats while the sun dipped below the horizon, casting a golden glow over the water.}
\end{itemize}

\subsection*{\texttt{gpt-4o-mini} / RIV}

\noindent Suffix: \textit{``We watched the current carry our little boat downstream until it disappeared.''}
\begin{itemize}[itemsep=2pt]
\item[\#5:] \textit{Their laughter echoed as they splashed in the shallow water, while we secured our belongings and prepared to set off on our adventure.}
\item[\#6:] \textit{They were gathering their gear, laughter bubbling as they prepared for the evening’s fishing trip. Excitement filled the air as they invited us to join them, sharing tales of past catches and the best spots along the river.}
\item[\#8:] \textit{They were preparing to set off on their own adventure, laughter echoing over the gentle ripples of the water.}
\item[\#9:] \textit{The laughter and chatter filled the air as we approached, the warm glow of the setting sun reflecting on the water's surface. As we got closer, we could see the familiar faces of our friends gathered around, ready to embark on the evening's adventure.}
\end{itemize}